\documentclass[journal]{IEEEtran}
\IEEEoverridecommandlockouts
\usepackage{cite}
\usepackage{amsmath,amssymb,amsfonts}
\usepackage{amsthm,bm}
\usepackage{graphicx}
\usepackage{textcomp}

\def\BibTeX{{\rm B\kern-.05em{\sc i\kern-.025em b}\kern-.08em
    T\kern-.1667em\lower.7ex\hbox{E}\kern-.125emX}}
\newcommand\numberthis{\addtocounter{equation}{1}\tag{\theequation}}
\usepackage{stackengine}
\usepackage{algorithm}
\usepackage{booktabs}
\usepackage{multicol}
\usepackage[noend]{algpseudocode}
\usepackage{subfig}
\usepackage[table,xcdraw]{xcolor}
\makeatletter
\def\BState{\State\hskip-\ALG@thistlm}
\makeatother

\algnewcommand\algorithmicforeach{\textbf{for each}}
\algdef{S}[FOR]{ForEach}[1]{\algorithmicforeach\ #1\ \algorithmicdo}
    
\begin{document}

\title{\huge Towards Federated Learning in UAV-Enabled Internet of Vehicles: A Multi-Dimensional Contract-Matching Approach \\
}

\author{
    Wei Yang Bryan Lim\thanks{WYB.~Lim is with Alibaba Group and Alibaba-NTU Joint Research Institute, Nanyang Technological University, Singapore. }, 
    Jianqiang Huang,\thanks{JQ.~Huang and XS.~Hua are with Alibaba Group. }
    Zehui Xiong\thanks{Z.~Xiong is with Alibaba-NTU Joint Research Institute, and also with School of Computer Science and Engineering, Nanyang Technological University, Singapore. }, 
      Jiawen Kang,\thanks{J.~Kang and D.~Niyato are with School of Computer Science and Engineering, Nanyang Technological University, Singapore. }
    Dusit Niyato,~\textit{IEEE Fellow}, 
    Xian-Sheng~Hua,~\textit{IEEE Fellow}, 
    Cyril Leung,\thanks{C. Leung is with The University of British Columbia and Joint NTU-UBC Research Centre of Excellence in Active Living for the Elderly (LILY).}
    Chunyan Miao \thanks{C.~Miao is with Joint NTU-UBC Research Centre of Excellence in Active Living for the Elderly (LILY) and School of Computer Science and Engineering, Nanyang Technological University, Singapore.}
    
   }

% Email: limw0201@e.ntu.edu.sg.
% Email: jianqiang.hjq@alibaba-inc.com, xiansheng.hxs@alibaba-inc.com.
% E-mails: zxiong002@e.ntu.edu.sg. E-mails: , dniyato@ntu.edu.sg.  E-mails: ascymiao@ntu.edu.sg, dniyato@ntu.edu.sg.

\maketitle

\begin{abstract}

Coupled with the rise of Deep Learning, the wealth of data and enhanced computation capabilities of Internet of Vehicles (IoV) components enable effective Artificial Intelligence (AI) based models to be built. Beyond ground data sources, Unmanned Aerial Vehicles (UAVs) based service providers for data collection and AI model training, i.e., Drones-as-a-Service, is increasingly popular in recent years. However, the stringent regulations governing data privacy potentially impedes data sharing across independently owned UAVs. To this end, we propose the adoption of a Federated Learning (FL) based approach to enable privacy-preserving collaborative Machine Learning across a federation of independent DaaS providers for the development of IoV applications, e.g., for traffic prediction and car park occupancy management. Given the information asymmetry and incentive mismatches between the UAVs and model owners, we leverage on the self-revealing properties of a multi-dimensional contract to ensure truthful reporting of the UAV types, while accounting for the multiple sources of heterogeneity, e.g., in sensing, computation, and transmission costs. Then, we adopt the Gale-Shapley algorithm to match the lowest cost UAV to each subregion. The simulation results validate the incentive compatibility of our contract design, and shows the efficiency of our matching, thus guaranteeing profit maximization for the model owner amid information asymmetry.
\end{abstract}


\newtheorem{definition}{Definition}
\newtheorem{lemma}{Lemma}
\newtheorem{theorem}{Theorem}

\newtheorem{property}{Property}

\section{Introduction}

%In traditional Vehicular Ad-Hoc Networks (VANETs), the vehicles are mainly regarded as communication nodes that facilitate transmission and data relay for efficient Vehicle-to-Vehicle (V2V) and Vehicle-to-Infrastructure (V2I) communications \cite{li2007routing,hartenstein2008tutorial}. 

Following the advancements in the Internet of Things (IoT) and edge computing paradigm, traditional Vehicular Ad-Hoc Networks (VANETs) that focus mainly on Vehicle-to-Vehicle (V2V) and Vehicle-to-Infrastructure (V2I) communications \cite{li2007routing,hartenstein2008tutorial} are gradually evolving into the Internet of Vehicles (IoV) paradigm \cite{xu2017internet,wan2016mobile}.

The IoV is an open and integrated network system which leverages on the enhanced sensing, communication, and computation capabilities of its component data sources, e.g., vehicular sensors, IoT devices, and Roadside Units (RSUs) \cite{yang2014overview}, to build data-driven applications for Intelligent Transport Systems, e.g., for traffic prediction \cite{wang2018internet}, traffic management \cite{kumar2018ant}, route planning \cite{florian2014privacy}, and other smart city applications \cite{ang2018deployment}. Coupled with the rise of Deep Learning, the wealth of data and enhanced computation capabilities of IoV components enable effective Artificial Intelligence (AI) based models to be built. 

Beyond ground data sources, aerial platforms are increasingly important today given that modern day traffic networks have grown in complexity. In particular, Unmanned Aerial Vehicles (UAVs) are  commonly used today to provide data collection and computation offloading support in the IoV paradigm. The UAVs feature the benefits of high mobility, flexible deployment, cost effectiveness \cite{zhou2014efficient}, and can also provide a more comprehensive coverage as compared to ground users. UAVs can be deployed, e.g., to capture images of car parks for the management and analysis of parking occupancy \cite{zhou2017car}, to capture images of roads and highways for traffic monitoring applications \cite{elloumi2018monitoring,coifman2006roadway,ke2018real}, and also to aggregate data from stationary vehicles and roadside units that in turn collect data of other passing vehicles periodically \cite{binol2018time}. Apart from data collection, the UAVs have also been used to provide computation offloading support for resource constrained IoV components \cite{zhang2018energy,bekkouche2018uavs}.

As such, studies proposing the Internet of Drones (IoD) and Drones-as-a-Service (DaaS) \cite{gharibi2016internet,koubaa2018dronetrack,koubaa2017service} have gained traction recently. Moreover, the DaaS industry is a rapidly growing one \cite{walia2019global} that comprises independent drone owners which provide on-demand data collection and model training for businesses and city planners.

Naturally, to build a better inference model, the independently owned UAV companies can collaborate by sharing their data collected from various sources, e.g., carparks, RSUs, and highways, for collaborative model training. However, in recent years, the regulations governing data privacy, e.g., General Data Protection Regulation (GDPR) are increasingly stringent. As such, this can potentially prevent the sharing of data across DaaS providers. To this end, we propose the adoption of a Federated Learning (FL) based \cite{mcmahan2016communication} approach to enable privacy-preserving collaborative ML across a federation of independent DaaS providers. 

In our system model (Fig. \ref{fig:sysmodel}), a client, hereinafter model owner, is interested in collecting data from a region for model training, e.g., for traffic prediction. Given the energy constraints of UAVs \cite{zeng2017energy}, the region is further divided into smaller subregions. The model owner then announces an FL task, e.g., the capturing of real-time traffic flow over highways or the collection of data from RSUs for model training \cite{ke2018real}. Then, only the DaaS providers, hereinafter UAVs, that are able to complete the task within the stipulated time and energy constraints respond to the model owner. Thereafter, the model owner assigns an optimal UAV to each subregion. After the UAV collects the sensing data, model training takes place on each UAV separately, following which only the updated model parameters are transmitted to the model owner for global aggregation.

\begin{figure}
	\includegraphics[width=\columnwidth]{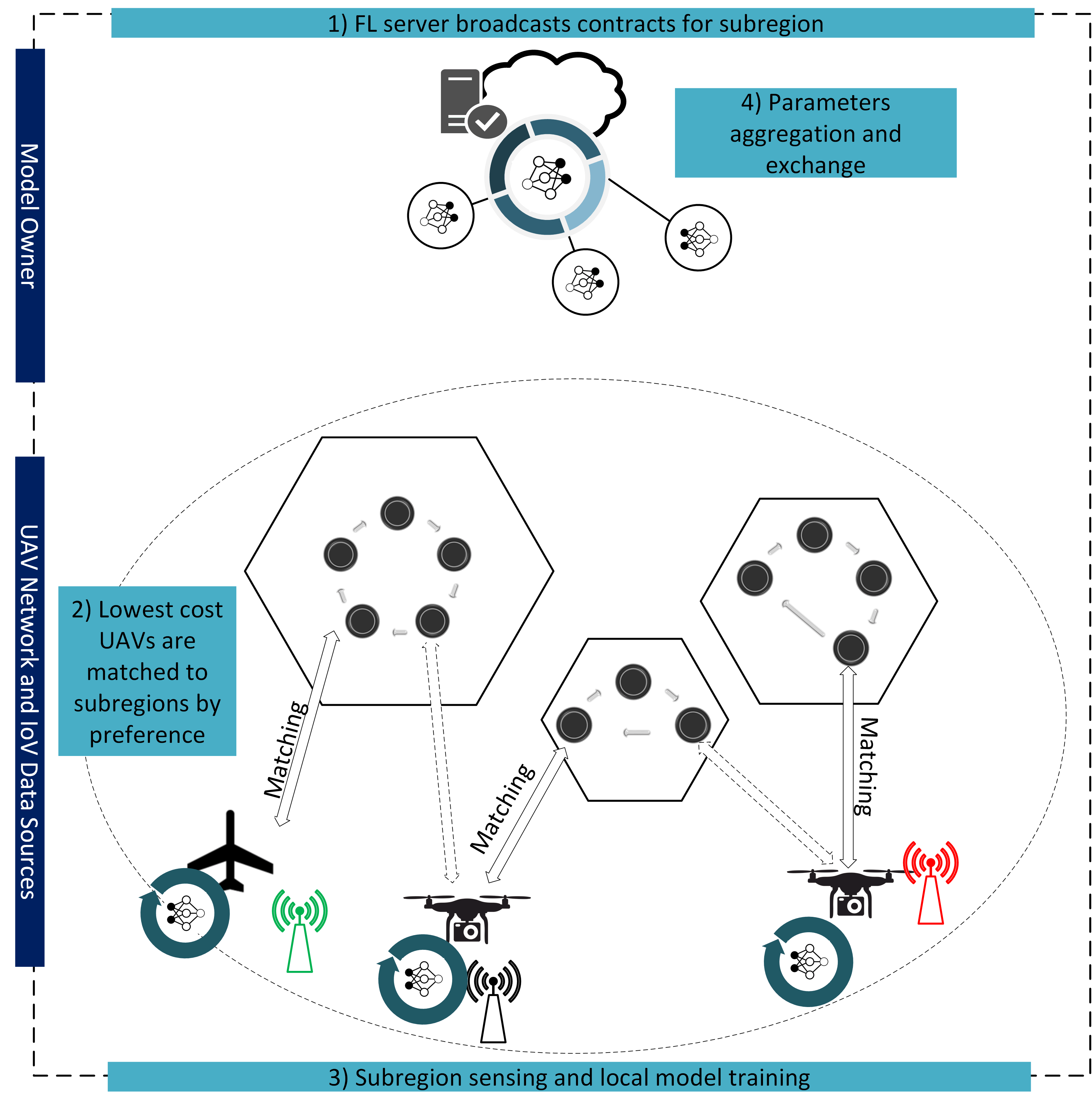}\par 
	\caption{Our proposed system model involving UAV-subregion contract-matching, and FL based collaborative learning within a federation of multiple UAVs. Note that each hexagon indicates a subregion, and within the subregion are nodes, e.g., RSUs, to visit as stipulated by the model owner.}
	\label{fig:sysmodel}
\end{figure}

Our proposed approach has three advantages. Firstly, the resource constrained IoV components are aided by the UAV deployment for completion of time sensitive sensing and model training tasks. Secondly, it preserves the privacy of the UAV-collected data through eliminating the need of data sharing across UAVs. Thirdly, it is communication efficient. The reason is that traditional methods of data sharing will require the raw data to be uploaded to an aggregating cloud server. With FL, only the model parameters need to be transmitted by the UAVs.

However, there exists an incentive mismatch between the model owner and the UAVs. On one hand, the model owners aim to maximize their profits by selecting the optimal UAVs which can complete the stipulated task at the lowest cost, e.g., in terms of sensing, transmission, and computation costs. On the other hand, the UAVs can take advantage of the information asymmetry and misreport their types so as to seek higher compensation. To that end, we leverage on the self-revealing properties of contract theory \cite{bolton2005contract} as an incentive mechanism design to appropriately reward the UAVs based on their actual types. In particular, given the complexity of the sensing and collaborative learning task, we consider a multi-dimensional contract to account for the multi-dimensional sources of heterogeneity in terms of UAV sensing, learning, and transmission capabilities.

After deriving optimal contracts to which the UAVs respond to, the possibility that multiple UAVs prefer a particular subregion still remains. To that end, we leverage on the Gale-Shapley (GS) \cite{dubins1981machiavelli} matching-based algorithm to assign the optimal UAVs to each subregion.

The contribution of this paper is as follows:

\begin{itemize}

\item We propose an FL based sensing and collaborative learning scheme in which UAVs collect the data and participate in privacy-preserving collaborative model training for applications in the IoV paradigm towards the development of an Intelligent Transport System.

\item In consideration of the incentive mismatches and information asymmetry between the UAVs and model owner, we propose a multi-dimensional contract-matching based incentive mechanism design that aims to leverage on the self-revealing properties of an optimal contract, such that the most optimal UAV can be matched to a subregion. 

\item Our incentive mechanism design considers a general UAV sensing, computation, and transmission model, and thus can be extended to specific FL based applications in the IoV paradigm.

\end{itemize}

The organization of this paper is as follows. Section \ref{sec:rw} reviews the related works, Section \ref{sec:sys} introduces the system model and problem formulation, Section \ref{sec:contract} discusses the multi-dimensional contract formulation, Section \ref{sec:matching} considers a matching-based UAV-subregion assignment, Section \ref{sec:pe} presents the performance evaluation of our proposed incentive mechanism design, and Section \ref{sec:conclude} concludes.
 
\section{Related Work}
\label{sec:rw}

In recent years, given the rising popularity of UAVs, there is an increasing number of UAV-related studies in the literature. One group of studies focus on the fundamental issues related to the challenges of UAV deployment, e.g., trajectory optimization \cite{zeng2017energy,oleynikova2016continuous,zhang2018cellular}, communication constraints \cite{abdulla2014optimal,alemayehu2017efficient}, as well as the efficient assignment and deployment of UAVs \cite{boccardo2015uav,brust2015networked}. Another group of studies propose specific applications of UAVs, e.g., as flying base stations \cite{alzenad20173}, with mobile cloudlets for computation offloading \cite{jeong2017mobile}, and for search and rescue missions \cite{scherer2015autonomous}. In particular, the UAVs are also increasingly considered for providing sensing services, i.e., data collection,  \cite{zhou2018mobile} and for the development of IoV related applications, e.g., for traffic prediction \cite{elloumi2018monitoring}, localization of ground vehicles \cite{liu2018uav}, and to facilitate vehicular communications \cite{zhuang2019sdn}.

The market of UAVs as service providers, e.g., in on-demand data collection, is a rapidly growing one \cite{walia2019global}. Given the heterogeneity in UAV types, e.g., in energy constraints and computation capabilities, the incentive mechanism design for UAV systems is an important issue. The study in \cite{ma2019strategic} adopts a game theoretic approach to analyze the offloading decisions of UAVs acting as flying cloudlets for IoT devices. In contrast, the study in \cite{zhang2018predictive} proposes the contract-theoretic approach to incentivize UAV base stations to contribute higher transmit power for enhanced coverage over wireless networks. In consideration of the limited availability of mobile charging stations for UAVs, the study in \cite{shin2019auction} proposes an auction-based approach to efficiently assign the UAVs to specific charging time slots so as to reduce congestion. 

However, given the  nascent field of FL, there are relatively few works that propose FL based collaborative learning schemes involving UAVs. To the best of our knowledge, the study of \cite{zeng2020federated} is the first to propose the implementation of FL for joint power allocation and scheduling of UAV swarms. With the increasingly stringent regulations related to data privacy, the adoption of FL can facilitate collaborative learning for the development of effective AI models, without the exchange of potentially sensitive raw data. As such, there is an urgent need to consider the incentive mechanism design for FL in UAV networks. 

To that end, we can take reference from the growing literature related to incentive mechanism design for FL. For example, the study in \cite{kang2019incentive} adopts a contract-theoretic approach to motivate workers to contribute more computation resource for efficient FL. On the other hand, the study in \cite{feng2019joint} formulates the Stackelberg game to analyze the inefficiency in model update transfer. As an extension, the study in \cite{zhan2020learning} uses a Stackelberg game formulation together with Deep Reinforcement Learning to design a learning-based incentive mechanism for FL. For a comprehensive survey in this area, we refer the readers to \cite{lim2019federated}. 

Apart from the traditional considerations of incentive design in FL, the UAV systems involve other sources of heterogeneity in UAV types, e.g., traversal costs. As such, the multi-dimensional sources of heterogeneity in UAVs have inspired us to adopt the multi-dimensional contract theoretic approach \cite{wang2019multi} in our incentive mechanism design. Moreover, in contrast to traditional works in contract theoretic mechanism design, our system model only involves the matching of a single, optimal UAV type to each subregion. This necessitates the use of the matching-based algorithm such as the GS algorithm. The use of matching for UAVs to subregions have also been studied in \cite{zhou2018mobile}. However, \cite{zhou2018mobile}  does not have any mechanism in place to ensure truthful reporting, while we leverage on the self-revealing properties of contract theory to that end. While the study of contract-matching has also been explored for resource allocation in vehicular fog computing \cite{zhou2019computation}, the contract considered is single-dimensional with simpler considerations. 

In summary, our study considers the adoption of FL to facilitate privacy preserving sensing and collaborative learning in the UAV services market, and proposes a multi-dimensional contract-matching design that aims to match the most optimal UAV to each sensing subregion, while accounting for the multiple sources of heterogeneity in UAV types.

\section{System Model and Problem Formulation}
\label{sec:sys}

We consider a network in which a model owner aims to collect data from stipulated nodes, e.g., from RSUs or images of segments in the highway, in a target sensing region to fulfill a time-sensitive task. One UAV is selected by the task publisher to cover each of the  subregions. Given information asymmetry and the multiple sources of heterogeneity in UAV cost types, the model owner leverages on the self-revealing properties of a multi-dimensional contract theoretic approach to choose one UAV suited to cover each of the subregion. After data collection, the UAV returns to their respective UAV bases for Federated Learning (FL) based model training.

Following \cite{zhou2018mobile}, the target sensing region can be modeled as a graph and divided into $N$ smaller graphs, i.e., subregions whose set is denoted $\mathcal{N} = \{1, \ldots, n, \ldots, N\}$, e.g., through the multilevel graph partition algorithm \cite{karypis1995multilevel}. The set of nodes in subregion $n$ is denoted $\mathcal{I}_n= \{I_1, \ldots, I_n, \ldots, I_N \}$ with the node $i$ in subregion $n$ located at $\bm{x}^n_i\in \mathbb{R}^3$. The Euclidean distance between two nodes $i$ and $i'$ located within subregion $n$, $\forall i, i' \in \mathcal{I}_n, i \neq i'$ is expressed as $l^n_{i,i'}$ where $l^n_{i,i'}=\vert\vert \bm{x}^n_i -\bm{x}^n_{i'}\vert\vert < \infty$, i.e., all nodes are inter-accessible. 

%Without loss of generality, all subregions are of equal importance to the model owner and they contain the same amount of data $D$ to be collected. For example, in a sensing task for smart agriculture, the agricultural field may be divided into two subregions in which both are of equal importance. Moreover, our work can be extended to alternative scenarios by altering the sensing rewards for each subregion based on its importance to the model owner.

A set $\mathcal{J}= \{1, \ldots, j, \ldots, J\}$ of $J$ unmanned aerial vehicles (UAVs) are located at bases situated around the target sensing region. Without loss of generality, we assume that each base owns a single UAV and  $J\geq N$. Moreover, our model can be easily extended to scenarios in which a UAV swarm is required for sensing in each subregion. Denote $C_j = \{C_1, \ldots, C_j, \ldots, C_J\}$ as the set of bases where $C_j$ refers to the base of UAV $j$ located at $\bm{y}_{C_j}\in \mathbb{R}^3$. The Euclidean distance between the base of UAV $j$ and subregion $n$ is expressed as $l^n_{C_j}$, where $l^n_{C_j} =\vert\vert \bm{y}_{C_j} -\bm{x}^n_{\bar{i}}\vert\vert  < \infty$ and $\bar{i}$ denotes the centre of the subregion. 

%In particular, $i^*$ refers to the node closest to a base, e.g., $i^* = \arg \min_{i} l^n_{C_j,i}$ for $C_j$. To perform a sensing task, the UAV first departs for node $i^*$, then, it covers its desired number of nodes before returning to the base for model training.

There are two stages in our system model as follows:

\begin{enumerate}

\item \textbf{Multi-Dimensional Contract Design:} The UAV types, e.g., sensing, traversal, and transmission costs, are private information not known to the model owner. As such, the model owner designs a multi-dimensional contract to leverage on the self-revealing mechanism so as to select the optimal UAV to cover each subregion. In particular, the model owner can maximize its profits by choosing the lowest cost UAV among all feasible UAVs that can complete the task within the time constraint.

%\item \textbf{UAV Node Coverage Optimization:} For \textit{each} subregion, the UAV chooses the optimal contract that reflects its type, i.e., the sensing, traversal, computation, and transmission cost types. We further discuss the varying UAV types in Section \ref{sec:uavsensing} to \ref{sec:uavtransmission}. Note that each UAV's preference can vary across different subregions. As an illustration, we consider a representative UAV $j$ with base $C_j$ and two subregions $n$ and $n^{\prime}$ where $l^{n^\prime}_{C_j,i^*}\gg l^n_{C_j,i^*}$. In this case, UAV $j$ has to traverse a longer distance to reach subregion $n^{\prime}$. As such, it is able to cover a smaller proportion of subregion $n^{\prime}$ relative to $n$ due to energy and time constraints. 

\item \textbf{UAV-Subregion Assignment:} Each UAV reports its type and ranks the $N$ subregions based on its preferences of coverage to the FL model owner. Then, a stable UAV-subregion matching is derived using the Gale-Shapley (GS) algorithm. Note that each UAV's preference can vary across different subregions. As an illustration, we consider a representative UAV $j$ with base $C_j$ and two subregions $n$ and $n^{\prime}$ where $l^{n^\prime}_{C_j}\gg l^n_{C_j}, \forall i_n \in I_n \text{ and } \forall i_{n^\prime} \in I_{n^\prime}$. In this case, UAV $j$ has to traverse a longer distance to reach subregion $n^{\prime}$. As such, it is able to cover a smaller proportion of subregion $n^{\prime}$ relative to $n$ due to energy and time constraints. 

\end{enumerate}

 In the following, we consider the sensing, computation, and data transmission model of a representative UAV.

\subsection{UAV Sensing Model}
\label{sec:uavsensing}

\begin{table}[]
\centering
\caption{Table of commonly used notations.}\label{tab:tablenotation}
\begin{tabular}{|l|l|}

\hline
\rowcolor[HTML]{C0C0C0} 
\textbf{Notation} & \textbf{Description}                                                                                           \\ \hline
              $n$    & Subregion                                                                                        \\ \hline
                 $j$ &   UAV                                                                            \\ \hline
                 $C_j$ & Base of UAV $j$                                                                                         \\ \hline
                  $l^n_j$ & Total sensing distance                                                                                      \\ \hline
                  $l^n_{C_j}$ & Total traversal distance                                                                                     \\ \hline
 $\tau_P^{j,n}$ & Total duration taken for traversal and sensing                                                                                    \\ \hline
 $ E_P^{j,n}$ & Total energy taken for traversal and sensing                                                                                    \\ \hline
 $ \alpha^n_j$ & Marginal cost of node coverage for sensing                                                                                    \\ \hline
 $ \psi^n_j$ & Traversal cost                                                                                    \\ \hline
 $\tau_C^{j,n}$ & Local computation duration                                                                                  \\ \hline
 $E^{j,n}_C$ & Total energy taken for computation                                                                                  \\ \hline
  $\beta_j$ & Marginal cost of node coverage for computation                                                                                \\ \hline
  $\tau_T^{j,n}$ & Total duration for transmission                                                                               \\ \hline
    $\zeta^n_j$ & Energy taken for transmission                                                                            \\ \hline
      $u^n_j$ & UAV utility                                                                            \\ \hline
      $R^n_j$ & Contractual rewards                                                                            \\ \hline
      $\phi$ & Unit cost of energy for the UAV                                                                          \\ \hline
 $\Pi$ & Model owner profit                                                                        \\ \hline
  $\Omega^n,\omega^n$ & Contract set and individual contract                                                                      \\ \hline
    $\tilde{R}$ & Compensation for sensing and computation costs                                                                        \\ \hline
  $\hat{R}$ & Compensation for traversal and transmission costs                                                                        \\ \hline
  $\upsilon(\alpha_y,\beta_z)$ & Marginal cost of node coverage                                                                       \\ \hline
  $\Phi_i$  & UAV auxiliary type \\ \hline
\end{tabular}
\end{table}

We consider a representative UAV $j$ tasked by the model owner to cover a proportion of nodes in the subregion $n$. Denote the node coverage assignment of UAV $j$ in subregion $n$ to be $\mathcal{A}^{j, n}=\left\{a_{i, i^{\prime}}^{j, n} | \forall i, i^{\prime} \in \mathcal{I}_{n}, i \neq i^{\prime}\right\}$ where $a^{j,n}_{i,i'}=1$ represents that the UAV has to fly through the segment between nodes $i$ and $i'$, and $a^{j,n}_{i,i'}=0$ implies otherwise.

%The UAV $j$ may not choose to cover all nodes in the subregion due to time and energy constraints. Denote the route selection strategy of UAV $j$ in subregion $n$ to be $\mathcal{A}^{j, n}=\left\{a_{i, i^{\prime}}^{j, n} | \forall i, i^{\prime} \in \mathcal{I}_{n}, i \neq i^{\prime}\right\}$ where $a^{j,n}_{i,i'}=1$ represents that the UAV will fly through the segment between nodes $i$ and $i'$, whereas $a^{j,n}_{i,i'}=0$ implies otherwise.

%set of preferred route coverage of UAV $j$ in region $S_n$ to be $\mathcal{P}^{n}_j$ and proportion of route coverage as $\theta^{n}_j= \frac{\vert\mathcal{P}^{n}_j\vert}{\vert\mathcal{I}_n\vert}$, where $\vert . \vert$ indicates cardinality.

The total distance $l^n_j$ traveled for sensing by UAV $j$ under assignment $\mathcal{A}^{j, n}$ is as follows:
\begin{equation}
l^n_j= \sum_{i' \neq i, i' \in \mathcal{I}_n}a^{j,n}_{i,i'} l^n_{i,i'}.
\numberthis
\label{eqn:uavdistance}
\end{equation}
Denote $\theta^{n}_j =  \frac{\sum_{i' \neq i, i' \in \mathcal{I}_n}a^{j,n}_{i,i'}}{|\mathcal{I}_n|}$ where $\vert \cdot \vert$ indicates cardinality, i.e., $\theta^{n}_j$ refers to the proportion of node coverage by UAV $j$ in subregion $n$ where $0 \leq \theta^{n}_j \leq 1$.

Apart from traveling between the nodes, the UAV has to travel to and from its base. Denote the total distance traveled by the UAV as  $L^n_j= l^n_j +
 l^n_{C_j}$. Hereinafter, we refer to $l^n_j$ as the \textit{sensing} distance, whereas $ l^n_{C_j}$ refers to the \textit{traversal} distance.

%
% l^n_{C_j,\tilde{i}} + l^n_{C_j,\bar{i}}$ where  $l^n_{C_j,\tilde{i}}$ and $l^n_{C_j,\bar{i}}$ refers to the distance traveled to and from the base respectively. Hereinafter, we refer to $l^n_{C_j,\tilde{i}} + l^n_{C_j,\bar{i}}$ as the \textit{traversal} distance, whereas $l^n_j$ refers to the \textit{sensing} distance.

% $l^n_{C_j,i^*}$ is the distance covered when the UAV travels from the base to the closest sensing node, and

% $L^n_j$ is the summation of distance between all nodes in region $n$, i.e., $L^n_j = \sum_{i' \neq i, i' \in \mathcal{I}_n} l^n_{i,i'} $

%$L^n_j$ is the total distance covered by UAV $j$ in region $n$ if it covers all nodes, i.e., $L^n_j = \sum_{i' \neq i, i' \in \mathcal{I}_n} l^n_{i,i'} + 2 l^n_{C_j,i^*}$ and $l^n_j \leq L^n_j $.

Following the works of \cite{zeng2019energy,zhang2018predictive}, each UAV travels with an average velocity $v_j$ and expends a fixed propulsion power $p_j =c_{j,1}v^3_{j} + \frac{c_{j,2}}{v_j}$ throughout the task for tractability, where $c_{j,1}$ and $c_{j,2}$ refers to the required power to balance the parasitic drag caused by skin friction and required power to balance the drag force of air redirection respectively\footnote{In practice, the propulsion power is in turn a function of other factors, e.g., reference area of the UAV and wing aspect ratio and weight. For simplicity, we consider that $p_j$ accounts for these factors.}. Note that the propulsion power consumed by the UAV when it changes its direction is negligible \cite{zeng2017energy}. The total duration taken for traversal and sensing is denoted $\tau_P^{j,n}=\frac{L^n_j}{ v_j}$, whereas the total energy consumed to cover the traversal and sensing distance is as follows:
\begin{flalign*}
 E_P^{j,n} = \frac{L^n_j}{ v_j} p_j 
 & = \frac { \theta^{n}_j l^n +l^n_{C_j}  }{v_j} p_j \\
 & = \frac{p_j l^n}{v_j} \theta^{n}_j + \frac{l^n_{C_j}}{v_j}p_j \\
 & = \alpha^n_j \theta^{n}_j + \psi^{n}_j,
\numberthis
\label{eqn:uavenergy}
\end{flalign*}
where  $l^n$ is the distance traveled by the UAV if it covers all nodes, i.e., $\theta^n_j=1$, $\alpha^n_j =  \frac{p_j l^n}{v_j} $ and $\psi^n_j = \frac{l^n_{C_j}}{v_j}p_j$ for notation simplicity. Note that $\alpha^n_j$  represents the sensing cost, i.e., marginal cost of node coverage for sensing in the subregion, whereas $\psi^{n}_j$ refers to the traversal cost, i.e., the energy cost of traveling to and from the base. A higher $\alpha^n_j$ can imply that the UAV $j$ requires greater propulsion power to complete the task, e.g., due to its larger weight or wing-aspect ratio, whereas a higher $\psi^{n}_j$ implies either a greater propulsion power to move, or a greater traversal cost, i.e., the subregion is farther away from the base. While the value of $\alpha^n_j$ varies across subregions due to the varying $l_n$, i.e., the marginal cost of node coverage varies according to the sensing area of the subregion, the \textit{ordering} of the UAV types based on the sensing costs is retained. On the other hand, the order of UAVs by traversal costs varies across subregions, based on the distance between the UAV base and each of the subregions.

% For simplicity, we use $\alpha_j$ in place of $\alpha^n_j$. 

\subsection{UAV Computation Model}

After the UAV $j$ covers its assigned set of nodes following assignment $\mathcal{A}^{j, n}$, it returns to the base $C_j$ for an FL based model training over $K$ global iterations where $\mathcal{K}=\{1,\ldots,k,\ldots,K\}$ to minimize the global loss $F^K\left(\boldsymbol{w}\right)$. Following \cite{konevcny2016federated}, each training iteration $k$ consists of three steps namely: (i) \textit{Local Computation}, i.e., the UAV trains the received global model $\boldsymbol{w}^{(k)}$ locally using the sensing data, (ii) \textit{Wireless Transmission}, i.e., the UAV transmits the model parameter update $\bm{h}^{(k)}_j$ to the model owner, and (iii) \textit{Global Model Parameter Update}, i.e., all parameter updates derived from the $N$ subregions are aggregated to derive an updated global model $\boldsymbol{w}^{(k+1)}$, where $\boldsymbol{w}^{(k+1)}=\cup_{j\in\mathcal{N}}(\boldsymbol{w}^{(k)}_{j}+\bm{h}^{(k)}_j)$, which is then transmitted back to the UAVs for the $(k+1)^{th}$ training iteration.

In general, a series of local model training is performed by the UAV to minimize an $L$-Lipschitz and $\gamma$-strongly convex local loss function $G_j$ up to the target accuracy $A^*$ defined by the model owner to derive the parameter update. Note that a larger value of $A^*$ implies greater deviation from the optimal value. Moreover, $0<A^*<1$, i.e., the local solution $\bm{h}^{(k)}_j$ does not have to be trained to optimality, e.g., to reduce local computation duration especially for time sensitive tasks. In particular, following the formulation in \cite{yang2019energy}:
\begin{flalign}
\nonumber
& G_{j}\left(\boldsymbol{w}^{(k)}, \boldsymbol{h}_{j}^{(k)}\right)-G_{j}\left(\boldsymbol{w}^{(k)}, \boldsymbol{h}_{j}^{(k) *}\right)  \\
& \quad\quad\quad \leq A^*\left(G_{j}\left(\boldsymbol{w}^{(k)}, \mathbf{0}\right)-G_{j}\left(\boldsymbol{w}^{(k)}, \boldsymbol{h}_{j}^{(k) *}\right)\right).
\end{flalign}

%As an illustration, the local parameter for UAV $j$ at iteration $k$ is $\bm w^{(k)}+\bm{h}^{(k)}_j$ where $\bm w^{(k)}$ denotes global parameter transmitted from the model owner to selected UAVs at iteration $k$.

The FL training is completed after $K=\frac{a}{1-A^*}$ global iterations where $a = \frac{2L^2}{\gamma ^2 \xi}$ and $0 \leq \xi \leq \frac{\gamma}{L}$. The total local computation duration $\tau_C^{j,n}$ is as follows:

\begin{flalign*}
\tau_C^{j,n}=K \left( \frac{V C_{j} \theta^{n}_{j} D^n \log _{2}(1 /A^*)}{f_{j}} \right),
\numberthis
\label{eqn:computationtime}
\end{flalign*}
whereas the energy consumption of UAV $j$ for computation is as follows:
\begin{flalign*}
E^{j,n}_C=K \left(\kappa C_{j} \theta^{j,n} D^n V \log _{2}(1 /A^*) f_{j}^{2}\right) = \beta_j \theta^{n}_j.
\numberthis
\label{eqn:computationenergy}
\end{flalign*}
$\kappa$ is the effective switched capacitance that depends on the chip architecture \cite{mao2016dynamic}, $C_j$ is the cycles per bit for computing one sample data of UAV $j$, $ \theta^{j,n} D^n$ is the unit of data samples collected by UAV $j$, $V \log _{2}(1 /A^*)$ refers to the lower bound on number of local iterations required to achieve local accuracy $A^*$ \cite{yang2019energy} where $V=\frac{2}{(2-L \delta) \delta \gamma}$, and $f_j$ refers to the computation capacity of the UAV $j$, measured by CPU cycles per second. For ease of notation, we denote $\beta_j = \kappa KC_{j} D^n V \log _{2}(1 /A^*) f_{j}^{2}$, i.e., a higher $\beta_j$ implies greater energy cost for computation per additional node coverage. Similar to $\alpha^n_j$, the value of $\beta^n_j$ varies across subregion due to the different units of data samples available for computation. However, the ordering of the UAV types based on computation costs is retained.

\subsection{UAV Transmission Model}
\label{sec:uavtransmission}

After local computation, the wireless transmission takes place from the selected UAVs to the model owner. For simplicity, we denote the achievable rate of the UAV $j$ to be a product of its transmit power $\rho_j$ and a scaling factor $\lambda^n_j$ which covers other considerations, e.g., bandwidth allocation and channel gain. 

The total time $\tau_T^{j,n}$ taken by the UAV to upload parameter update $\bm{h}^{(k)}_j$ of size $H$ is as follows: $\tau_T^{j,n} = K\frac{H}{\lambda^n_j\rho_j}$. Note that the model upload size is constant regardless of the number of global iterations or quantity of data collected, given the fixed dimensions of the model update. The transmission energy consumption, denoted $\zeta^n_j$, is as follows:
\begin{equation}
E^{j,n}_T=\tau_T^{j,n}\rho_j =\zeta^n_j.
\numberthis
\label{eqn:transmissionenergy}
\end{equation}

%Following the studies in \cite{kang2019incentive,tran2019federated}, we also assume the UAVs face the similar wireless communication environments, i.e., $\lambda_j = \lambda \text{ } \forall j \in \mathcal{J}$. As such, for simplicity, we set the transmission energy consumption to be constant across all UAVs as follows:

%The number of global iterations $n$ is required as follows:
%\begin{flalign*}
%n = \frac{a}{1-\eta} \triangleq I_{0}
%\numberthis
%\label{eqn:globalfl}
%\end{flalign*}
%where $a=\frac{2 L^{2}}{\gamma^{2} \xi} \ln \frac{1}{\epsilon_{0}}$ and $\epsilon_{0}$ denotes target global accuracy, i.e., $F\left(\boldsymbol{w}^{(n)}\right)-F\left(\boldsymbol{w}^{*}\right) \leq \epsilon_{0}\left(F\left(\boldsymbol{w}^{(0)}\right)-F\left(\boldsymbol{w}^{*}\right)\right)$

%The number of local iterations $i$ is as follows:
%\begin{flalign*}
%i = v \log _{2}(1 / \eta)
%\numberthis
%\label{eqn:globalfl}
%\end{flalign*}
%where $v=\frac{2}{(2-L \delta) \delta \gamma}$ test

\subsection{UAV and Model Owner Utility Modeling}

The utility function of a representative UAV $j$ covering subregion $n$ can be expressed as follows:
\begin{flalign*}
 u^n_j  (\theta^{n}_j) &= R^n_j (\theta^{n}_j) - \phi\left(E_P^{j,n} + E^{j,n}_C + E^{j,n}_T \right) \\ 
&  =  R^n_j(\theta^{n}_j) -  \phi\left(\alpha^n_j \theta^{n}_j+ \psi^{n}_j + \beta^n_j \theta^{n}_j + \zeta^{n}_j \right),\\
\label{eqn:uavutility}
\numberthis
\end{flalign*}
where $R^n_j(\theta^{n}_j)$ refers to the contractual rewards and $\phi$ refers to the unit cost of energy.

%
%The type-$(x,y,z,q)$ UAV which covers $\theta(x,y,z,q)$ proportion of nodes in subregion $n$ has its utility expressed as follows, where $R=R(\theta(x,y,z,q))$ refers to contractual rewards to ease the notation:
%\begin{flalign*}
% u^n_{x,y,z} (\theta(x,y,z,q),R) 
% &=  R- \phi( \alpha_y \theta(x,y,z,q)  +  \\
% & \beta_z \psi_x+\theta(x,y,z,q)+ \zeta(x,y,z,q ) .
%\label{eqn:uavutility}
%\numberthis
%\end{flalign*}

Following \cite{liu2018edge,khan2019federated,zhan2020learning}, the FL model accuracy $\Upsilon(\sum^N_{n=1} \theta ^n_{j^*}D^n )$ is a concave function of the aggregate data collected across $N$ subregions by the $N$ selected UAVs. In particular, the inference accuracy of the model is improved when more nodes are covered, i.e., a model trained using data across a more comprehensive coverage of classes may be built. Without loss of generality, we consider the aggregate model performance to be an  average of node coverage across all regions, analogous to the Federated Averaging algorithm:
\begin{flalign}
\Upsilon \left(\sum^N_{n=1} \theta ^n_{j^*} D^n \right) = \frac{1}{N} \sum^N_{n=1} log(1+ \mu\theta ^n_{j^*}D^n ),
\label{eqn:accuracy}
\numberthis
\end{flalign}
where $\mu>0$ is the system parameter. The total profit obtained from all UAVs is thus as follows: 
\begin{flalign}
 \Pi(\Omega) =\sigma \Upsilon \left(\sum^N_{n=1} \theta ^n_{j^*} D^n \right) - \sum^N_{n=1}R^n_{j*}
\label{eqn:modelownerutility}
\numberthis
\end{flalign}
where $\sigma>0$ refers to the conversion parameter from model performance to profits, and the contractual reward expense for each selected UAV is denoted as $R^n_{j^*}$. In the next section, we devise the optimal contract which satisfies the Individual Rationality and Incentive Compatibility constraints.

\section{Multi-Dimensional Contract Design}
\label{sec:contract}

In this section, we first consider a multi-dimensional contract formulation. To solve the multi-dimensional contract, we sort the UAV types according to an auxiliary variable which reflects the marginal cost of node coverage. Then, we relax the constraints for contract feasibility and include a fixed compensation component for traversal and transmission costs so as to solve for the optimal contract.

\subsection{Contract Condition Analysis}

Given that the sensing cost $\alpha$, traversal cost $\psi$, computation cost $\beta$, and transmission cost $\zeta$ are all private information that are not precisely known by the model owner, we consider the multi-dimensional contract theoretic incentive mechanism design to leverage on its self-revealing properties. 

The UAVs can be classified into different types to characterize their heterogeneity. In particular, the UAVs can be categorized into a set $\Psi= \{\psi^n_x: 1 \leq x \leq X \}$ of $X$ traversal cost types, set $\mathcal{A}= \{\alpha^n_y: 1 \leq y \leq Y \}$ of $Y$ sensing cost types, set $\mathcal{B}= \{\beta^n_z: 1 \leq z \leq Z \}$ of $Z$ computation cost types, and set $\mathcal{C} = \{\zeta^n_q: 1 \leq q \leq Q \}$ of Q transmission cost types. 

Without loss of generality, we also assume that the user types are indexed in non-decreasing orders in all four dimensions: $0<\psi_1 \leq \psi_2 \leq \cdots \leq \psi_X$, $0<\alpha^n_1 \leq \alpha^n_2 \leq \cdots \leq \alpha^n_Y$, $0 < \beta^n_1 \leq \beta^n_2 \leq \cdots \leq \beta^n_Z $, and $0 < \zeta^n_1 \leq \zeta^n_2 \leq \cdots \leq \zeta^n_Q $. For ease of notation, we represent a UAV of traversal cost type $x$, sensing cost type $y$, computation cost type $z$, and transmission cost type $q$ to be that of type-$(x,y,z,q)$. 

To enforce the UAVs to truthfully reveal their private information, we adopt a two-step procedure for the contract design: 

\begin{enumerate} 

\item \textit{Multi-Dimensional Contract Design:} We convert the multi-dimensional problem into a single-dimensional contract formulation following the approach in \cite{wang2019multi}. In particular, we sort the UAVs by an auxiliary, one-dimensional type $\Phi(\alpha^n_y,\beta^n_z)$ in the ascending order based on the marginal cost of node coverage, i.e., sensing and computation cost types. Then, we solve for the optimal contract for each subregion $n$ denoted $\Omega^n (\mathcal{A},\mathcal{B}) = \{\omega^n_{y,z}: 1 \leq y \leq Y, 1 \leq z \leq Z  \}$ where $n \in \mathcal{N}$ to derive the optimal node coverage-contract reward bundle $\{\theta_{y,z}^n,\tilde{R}^n_{y,z}\}$.

\item \textit{Traversal Cost Compensation:} In contrast to existing works on multi-dimensional contracts, the UAVs also incur the additional traversal cost and transmission cost components, both of which are not coupled with the marginal cost of node coverage. In other words, these costs have to be incurred regardless of the number of nodes a UAV decides to cover in the subregion. For each contractual reward, we add in a fixed compensation $\hat{R}$ to derive the final contract bundle $\{\theta_{y,z}^n,(\tilde{R}^n_{y,z}+\hat{R})\}$.

\end{enumerate}

We first discuss the multi-dimensional contract formulation as follows. A contract is feasible only if the Individual Rationality (IR) and Incentive Compatibility (IC) constraints hold simultaneously.

\begin{definition}
\label{def:ir}
\normalfont
Individual Rationality (IR): Each type-$(y,z)$ UAV achieves non-negative utility if it chooses the contract item designed for its type, i.e., contract item $\omega_{y,z}$.
	\begin{align*}
	u_{y,z}(\omega_{y,z}) \geq 0 , 1 \leq y \leq Y, 1 \leq z \leq Z.\
	\label{eqn:ir}
	\numberthis
	\end{align*}
\end{definition}

\begin{definition}
\label{def:ic}
\normalfont
Incentive Compatibility (IC): Each type-$(y,z)$ UAV achieves the maximum utility if it chooses the contract item designed for its type, i.e., contract item $\omega_{y,z}$. As such, it has no incentive to choose contracts designed for other types.
	\begin{align*}
	u_{y,z}(\omega_{y,z}) \geq u_{,y,z}(\omega_{y',z'}), 
1 \leq y \leq Y, 
	1 \leq z \leq Z, \\
	 y \neq y^{\prime}, z \neq z^{\prime}.
	\label{eqn:ic}
	\numberthis
	\end{align*}
\end{definition}
The multi-dimensional contract formulation is as follows:
\begin{flalign*}
\label{eqn:contractoptimization}
& \max_{\Omega} \Pi(\Omega ^n (\mathcal{A},\mathcal{B}))\\ 
& \textrm{s.t. } (\ref{eqn:ir}),(\ref{eqn:ic}) .
\numberthis
\end{flalign*}
However, the optimization problem in (\ref{eqn:contractoptimization}) involves $YZ$, i.e., IR constraints and $YZ(YZ-1)$, i.e., IC constraints, all of which are non-convex. Therefore, we first convert the contract into a single-dimensional formulation in the next section.

%In contrast to the work of \cite{wang2019multi}, an additional cost component other than the marginal cost is involved. Note that the traversal cost is a \textit{sunk} cost, i.e., a cost that cannot be recovered once the UAV decides to travel to the subregion. In other words, the UAV only makes its node coverage decision based on marginal cost once it has decided to travel to a subregion. To ensure that the monotonicity constraint we subsequently discuss in Lemma \ref{lemma:monotonic} holds, we first sort by marginal cost of route coverage (Step 1). Then, within each $\phi$ type, the UAVs are further sorted by ascending order based on their $\psi$ values to account for variations in the $\Psi$ dimension. 

%we derive the \textit{unit cost} of node coverage $\eta(\psi_x,\alpha_y,\beta_z)$ through dividing both sides of (\ref{eqn:uavutility}) by $\theta(x,y,z)$. The revised utility $\tilde{u}_{x,y,z}$ of the type-$(x,y,z)$ UAV is as follows:
%\begin{flalign*}
% \tilde{u}_{x,y,z} (\theta(x,y,z),R)  
% =  \eta(\psi_x,\alpha_y,\beta_z)+R,
%\label{eqn:uavutilityrevised}
%\numberthis
%\end{flalign*}
%where $\eta(\psi_x,\alpha_y,\beta_z) = - \phi ( \alpha_y  + \frac{1}{\theta(x,y,z)}\psi_x +\beta_z + \frac{1}{\theta(x,y,z)} E^T) < 0$. Intuitively, the coverage of an additional node results in the expense of sensing and computation costs.  

%This involves a two-step procedure:

\subsection{Conversion Into A Single-Dimensional Contract}

In order to account for the marginal cost of node coverage, we consider a revised utility $\tilde{u}_{y,z}$ of the UAV type-$(y,z)$ that excludes the traversal and transmission costs as follows:
\begin{flalign*}
 \tilde{u}_{y,z} (\theta(y,z),\tilde{R}_{y,z}) 
 =  \eta(\alpha_y,\beta_z)+ \tilde{R}_{y,z},
\label{eqn:uavutility}
\numberthis
\end{flalign*}
where we denote $\eta(\alpha_y,\beta_z)= - \phi( \alpha_y \theta(y,z) + \beta_z \theta(y,z) )$ for ease of notation, and $\tilde{R}_{y,z}$ refers to the contractual reward arising from the multi-dimensional contract design. To focus on a representative contract, we drop the $n$ superscripts for now. Given that the ranking of marginal cost types does not change across subregion, note that our contract design is a general one applicable to all subregions. 

We derive the marginal cost of node coverage $\upsilon(\alpha_y,\beta_z)$ for the type-($y,z$) UAV as follows: 
\begin{flalign*}
\upsilon(\alpha_y,\beta_z) = -\frac{\partial \eta(\alpha_y,\beta_z)}{\partial \theta(y,z)} = \phi(\alpha_{y} + \beta_{z}).
\numberthis
\label{eqn:auxiliary}
\end{flalign*}
Intuitively, $\frac{\partial \eta(\alpha_y,\beta_z)}{\partial \theta(y,z)} < 0$ since the coverage of an additional node results in the additional expenses of sensing and computation costs. A larger value of $\upsilon(\alpha_y,\beta_z)$ implies a larger marginal cost of node coverage, due to the greater sensing and computation costs incurred for a particular UAV type.

We can now sort the $YZ$ UAVs according to their marginal cost of node coverage in a non-decreasing order as follows:
\begin{flalign*}
\Phi_1(\theta) , \Phi_2(\theta) , \ldots , \Phi_i(\theta) , \ldots , \Phi_{YZ}(\theta),
\label{eqn:auxiliary}
\numberthis
\end{flalign*}
where $\Phi_i(\theta)$ denotes the auxiliary type-$\Phi_i(\theta)$ user. Given the sorting order, the UAV types are in an ascending order based on their marginal cost of node coverage:
\begin{flalign*}
\upsilon(\theta,\Phi_1) \leq \upsilon(\theta,\Phi_2) \leq \cdots \leq \upsilon(\theta,\Phi_i) \leq \cdots \leq \upsilon(\theta,\Phi_{YZ}),
\label{eqn:auxiliary}
\numberthis
\end{flalign*}

Note that for ease of notation, we use type-$\Phi_i$ or type-$i$ interchangeably to represent the auxiliary type-$i$ user. In addition, we refer to $\eta_i (\theta_i)$ and $\eta (\theta_i,\Phi_i)$ interchangeably to represent the new ordering subsequently. Similarly, to represent the marginal cost of node coverage, we use $\upsilon_i (\theta_i)$ and $\upsilon (\theta_i,\Phi_i)$. In the next section, we derive the necessary and sufficient conditions for the contract design.

\subsection{Conditions For Contract Feasibility}

We derive the necessary conditions to guarantee contract feasibility based on the IR and IC constraints as follows.

\begin{lemma} 
\label{lemma:feasible}
\normalfont
For any feasible contract $\Omega \{ \mathcal{A},\mathcal{B} \}$, we have $\theta_i  < \theta_{i^\prime}$ if and only if $\tilde{R}_i < \tilde{R}_i'$, $i \neq i^\prime$.
\end{lemma}
\begin{proof}
We first prove the sufficiency, i.e., if $R_i < R_i' \Rightarrow \theta_i  < \theta_{i^\prime}$. From the IC constraint of type-$\Phi_i$ UAV we have:
\begin{align*}
	\eta (\theta_i,\Phi_i) + \tilde{R}_i \geq  \eta (\theta_i,\Phi_{i'}) + \tilde{R}_{i'},\\
	\eta (\theta_i,\Phi_i) - \eta (\theta_i,\Phi_{i'})  \geq    \tilde{R}_{i'}- \tilde{R}_i > 0,
	\numberthis
		\end{align*}
which implies:
\begin{align*}
	\eta (\theta_i,\Phi_i)\geq \eta (\theta_i,\Phi_{i'}),
	\numberthis
		\end{align*}
Given that $\frac{\partial \eta (\theta_i,\Phi_i)}{\partial \theta_i}<0$, we can deduce $\theta_i  < \theta_{i^\prime}$.

Next, we prove the necessity, i.e., $\theta_i  < \theta_{i^\prime}\Rightarrow \tilde{R}_i < \tilde{R}_i'$. Similarly, we consider the IC constraint of the type-$\Phi_i$ UAV:
\begin{align*}
	\eta (\theta_i,\Phi_i) + \tilde{R}_i \geq  \eta (\theta_{i'},\Phi_{i}) + \tilde{R}_{i'},\\
	\eta (\theta_i,\Phi_i) -  \eta (\theta_{i'},\Phi_{i}) \geq    \tilde{R}_{i'}- \tilde{R}_i.
	\numberthis
		\end{align*}
Given $\theta_i  < \theta_{i^\prime}$, we deduce $\eta (\theta_i,\Phi_i) <  \eta (\theta_{i'},\Phi_{i})$, which follows that $\tilde{R}_{i'} < \tilde{R}_i$. The proof is now completed.
\end{proof}

\begin{lemma} 
\label{lemma:monotonic}
\normalfont
Monotonicity: For any feasible contract $\Omega \{ \mathcal{A},\mathcal{B} \}$, if $\upsilon (\theta_i,\Phi_i)>\upsilon (\theta_i,\Phi_{i^\prime})$, it follows that $\theta_i \leq \theta_{i^\prime}$.
\end{lemma}
\begin{proof}
We adopt the proof by contradiction to validate the monotonicity condition. We first assume that there exists $\theta_i > \theta_{i^\prime}$ such that $\upsilon (\theta_i,\Phi_i)>\upsilon (\theta_i,\Phi_{i^\prime})$, i.e., the lemma is incorrect.

We consider the IC constraints for the type $\Phi_i$ and $\Phi_{i^\prime}$ UAV:
\begin{flalign*}
	\eta (\theta_i,\Phi_i) + \tilde{R}_i \geq  \eta (\theta_{i'},\Phi_{i}) + \tilde{R}_{i'},\\
	\eta (\theta_{i^\prime},\Phi_{i^\prime}) + \tilde{R}_{i^\prime} \geq  \eta (\theta_i,\Phi_{i'}) + \tilde{R}_{i}.
\end{flalign*}

Then, we add the constraints together and rearrange the terms to obtain:
\begin{flalign*}
	%\eta (\theta_i,\Phi_i) +\eta (\theta_{i^\prime},\Phi_{i^\prime}) \geq  \eta (\theta_i,\Phi_{i'})+ \eta (\theta_{i'},\Phi_{i}),\\
	 \left[\eta (\theta_i,\Phi_i)- \eta (\theta_{i'},\Phi_{i})\right] - \left[\eta (\theta_i,\Phi_{i'})- \eta (\theta_{i^\prime},\Phi_{i^\prime})\right] \geq 0.
	 \numberthis
	 \label{eqn:rearrange}
\end{flalign*}

By the fundamental theorem of calculus, we have:
\begin{flalign*}
	 &\quad\quad \left[\eta (\theta_i,\Phi_i)- \eta (\theta_{i'},\Phi_{i})\right] - \left[\eta (\theta_i,\Phi_{i'})- \eta (\theta_{i^\prime},\Phi_{i^\prime})\right]\\
	 &=\int_{\theta_{i'}}^{\theta_{i}} \frac{\partial \eta (\theta,\Phi_i)} {\partial \theta}d \theta-\int_{\theta_{i'}}^{\theta_{i}}  \frac{\partial \eta (\theta,\Phi_{i'})} {\partial \theta}d \theta \\
	 & = \int_{\theta_{i'}}^{\theta_{i}}\left[\frac{\partial \eta (\theta,\Phi_i)} {\partial \theta} -\frac{\partial \eta (\theta,\Phi_{i'})} {\partial \theta}   \right] d \theta\\
	 & = -\int_{\theta_{i'}}^{\theta_{i}}\left[\upsilon\left(\theta,\Phi_i \right) - \upsilon\left(\theta,\Phi_{i'} \right)  \right]d \theta.
	 \numberthis
	 \label{eqn:fund}
\end{flalign*}
Given (\ref{eqn:contractoptimization}), as well as the assumption $\theta_i > \theta_{i^\prime}$ and $\eta (\theta_i,\Phi_i)>\eta (\theta_i,\Phi_{i^\prime})$, we can deduce that (\ref{eqn:fund}) is negative, which contradicts with  (\ref{eqn:rearrange}). As such, there does not exist  $\theta_i > \theta_{i^\prime}$ and $\eta (\theta_i,\Phi_i)>\eta (\theta_i,\Phi_{i^\prime})$ for the feasible contract, which confirms that the lemma is correct. The proof is now completed. 
\end{proof}

As such, Lemmas  $1$ and $2$ give us the necessary conditions of the feasible contract in the following theorem.

\begin{theorem}
\normalfont
A feasible contract must meet the following conditions:
\label{monotonicitycondition}
\end{theorem}
\begin{equation}
\left\{\begin{array}{c}\theta_{1} \geq \theta_{2} \geq \cdots \geq \theta_{i} \geq \cdots \geq \theta_{YZ} \\ 
\tilde{R}_{1} \geq \tilde{R}_{2} \geq \cdots \geq \tilde{R}_{i} \geq \cdots \geq \tilde{R}_{YZ}\end{array}\right.
\numberthis
\end{equation}

Next, we further relax the IR and IC constraints. Due to the independence of $\Phi_i$ on the contract item $\{ \theta, \tilde{R}\}$, i.e., $\Phi_i(\theta,\tilde{R}) = \Phi_i(\theta^\prime,\tilde{R}^\prime)$, $\theta \neq \theta^\prime, \tilde{R} \neq \tilde{R}^\prime$, the UAV type does not change with the node coverage and contract rewards. In addition, the ordering of the type by marginal costs does not change with the subregion $n$. As such, we are able to deduce the minimum utility UAV $\Phi_{max}$, i.e., the UAV type that incurs the highest marginal cost of node coverage, as follows:
 \begin{flalign}
\Phi_{max}(\theta)=\arg \min _{\Phi_{i}} \tilde{u} \left(\theta,\tilde{R},\Phi_i \right).
 \end{flalign}
Intuitively $\Phi_{max}= \Phi_{YZ}$, i.e., the UAV characterized by $\{\alpha_Y,\beta_Z\}$ is the UAV which incurs the highest marginal cost of node coverage, and hence it is the minimum utility UAV.

\begin{lemma} 
\label{lemma:reduceir}
\normalfont
If the IR constraint of the minimum utility UAV type $\Phi_{YZ}$ is satisfied, the other IR constraints will also hold.
\end{lemma}
\begin{proof}
From the IC constraint and the sorting order $\eta_1(\theta_1) \geq \cdots \geq \eta_i(\theta_i) \cdots \geq \eta_{YZ}(\theta_{YZ})$, we have the following relation:
\begin{flalign*}
\eta_{i}(\theta_{i}) + \tilde{R}_{i} \geq \eta_{i}(\theta_{YZ}) + \tilde{R}_{YZ} \geq  \eta_{YZ}(\theta_{YZ}) + \tilde{R}_{YZ} \geq 0.
\end{flalign*}
As such, as long as the IR constraint of the UAV  type $\Phi_{YZ}$ is satisfied, it follows that the IR constraints of the other UAVs will also hold.
\end{proof}

\begin{lemma}
\label{lemma:reduceic}
\normalfont
For a feasible contract, if $\omega_{i-1} \stackrel{P I C}{\Leftrightarrow} \omega_{i} $ and $ \omega_{i} \stackrel{P I C}{\Leftrightarrow} \omega_{i+1}, $ then $\omega_{i-1} \stackrel{P I C}{\Leftrightarrow} \omega_{i+1}$.
\end{lemma}
Note that the relation $\omega_{i} \stackrel{P I C}{\Leftrightarrow} \omega_{i'}$, $i \neq i'$ implies the \textit{Pairwise Incentive Compatibility} (PIC), which is fulfilled under the following condition:
\begin{flalign*}
\left\{\begin{array}{l}
u_{i}\left(\omega_{i}\right) \geq u_{i}\left(\omega_{i'}\right), \\
u_{i'}\left(\omega_{i'}\right) \geq u_{i'}\left(\omega_{i}\right).
\end{array}\right.
\end{flalign*}
\begin{proof}
Suppose we have three UAV types $\Phi_{i-1}$, $\Phi_{i}$, and $\Phi_{i+1}$ where $i-1<i<i+1$. The Local Upward Incentive Constraint (LUIC), i.e., IC constraint between the $i^{th}$ and $({i+1})^{th}$ UAV is as follows:
\begin{flalign*}
\eta\left(\theta_{i-1}, \Phi_{i-1}\right)+\tilde{R}_{i-1} \geq \eta\left(\theta_{i}, \Phi_{i-1}\right)+\tilde{R}_{i} \\
\eta\left(\theta_{i}, \Phi_{i}\right)+\tilde{R}_{i} \geq \eta\left(\theta_{i+1}, \Phi_{i}\right)+\tilde{R}_{i+1}.
\numberthis
\label{eqn:luic}
\end{flalign*}

In addition, we consider:
\begin{flalign*}
& \left[\eta\left(\theta_{i+1}, \Phi_{i}\right)-\eta\left(\theta_{i}, \Phi_{i}\right)\right]-\left[\eta\left(\theta_{i+1}, \Phi_{i-1}\right)-\eta\left(\theta_{i}, \Phi_{i-1}\right)\right] \\
=& \int_{\theta_{i}}^{\theta_{i+1}} \frac{\partial \eta\left(\theta, \Phi_{i}\right)}{\partial \theta} d \theta-\int_{\theta_{i}}^{\theta_{i+1}} \frac{\partial \eta\left(\theta, \Phi_{i-1}\right)}{\partial \theta} d \theta \\
= & \int_{\theta_{i}}^{\theta_{i+1}} \frac{\partial \eta\left(\theta, \Phi_{i}\right)}{\partial \theta}-\frac{\partial \eta\left(\theta, \Phi_{i-1}\right)}{\partial \theta} d \theta \\
=&-\int_{\theta_{i}}^{\theta_{i+1}}\left[\upsilon\left(\theta, \Phi_{i}\right)-\upsilon\left(\theta, \Phi_{i-1}\right)\right] d \theta\\
=&\int_{\theta_{i+1}}^{\theta_{i}}\left[\upsilon\left(\theta, \Phi_{i}\right)-\upsilon\left(\theta, \Phi_{i-1}\right)\right] d \theta.
\numberthis
\label{eqn:luic2}
\end{flalign*}
Given the order of marginal cost of node coverage in (\ref{eqn:auxiliary}), it follows that (\ref{eqn:luic2}) is positive. As such, we have:
\begin{flalign*}
 \eta\left(\theta_{i+1}, \Phi_{i}\right)-\eta\left(\theta_{i}, \Phi_{i}\right) \geq \eta\left(\theta_{i+1}, \Phi_{i-1}\right)-\eta\left(\theta_{i}, \Phi_{i-1}\right) .
 \numberthis
 \label{eqn:luic3}
\end{flalign*}
Adding the LUIC inequalities presented in (\ref{eqn:luic}) together with that of (\ref{eqn:luic3}), we have:
\begin{flalign*}
\eta(\theta_{i-1},\Phi_{i-1}) + \tilde{R}_{i-1} \geq \eta(\theta_{i+1},\Phi_{i-1}) + \tilde{R}_{i+1}.
\numberthis
\label{eqn:luic4}
\end{flalign*}
By considering the Local Downward Incentive Constraint (LDIC), i.e., IC constraint between the $i^{th}$ and ${(i-1)}^{th}$ UAV, as well as adopting the approach in (\ref{eqn:luic2}), we are able to derive that:
\begin{flalign*}
\eta(\theta_{i+1},\Phi_{i+1}) + \tilde{R}_{i+1} \geq \eta(\theta_{i-1},\Phi_{i+1}) + \tilde{R}_{i-1}.
\numberthis
\label{eqn:luic5}
\end{flalign*}
As such, given that the LUIC and LDIC hold, we have proven that the PIC of the contracts hold, i.e., $\omega_{i-1} \stackrel{P I C}{\Leftrightarrow} \omega_{i+1}$.
\end{proof}

With Lemma \ref{lemma:reduceir}, we are able to reduce $YZ$ IR constraints into a single constraint, i.e., as long as the minimum utility type $\Phi_{YZ}$ UAV has a non-negative utility, it follows that the other IR constraints will hold. Moreover, with Lemma \ref{lemma:reduceic}, we are able to reduce $YZ(YZ-1)$ constraints into $YZ-1$ constraints, i.e., as long as the PIC constraint of the type $\Phi_i$ and type $\Phi_{i+1}$ UAV holds, it follows that the IC constraints between the type $\Phi_i$ and all other UAV types will hold.

With this, we are able to derive a tractable set of sufficient conditions for the feasible contract in Theorem \ref{theorem:sufficient} as follows. The first condition refers to the reduced IR condition corresponding to Lemma (\ref{lemma:reduceir}), whereas the second condition refers to the PIC condition between the type $\Phi_i$ and type $\Phi_{i+1}$ UAV corresponding to Lemma (\ref{lemma:reduceic}).
\begin{theorem}
\normalfont
A feasible contract must meet the following sufficient conditions:
\label{theorem:sufficient}
\end{theorem}
\begin{enumerate}
\item $\eta_{YZ}(\theta_{YZ}) + \tilde{R}_{YZ} \geq 0$
\item $\tilde{R}_{i+1}+ \eta(\theta_{i+1},\Phi_{i+1}) - \eta(\theta_{i},\Phi_{i+1})  \geq  \tilde{R}_{i}  \geq \tilde{R}_{i+1} + \eta(\theta_{i+1},\Phi_{i})-\eta(\theta_{i},\Phi_{i}) $.
\end{enumerate}

\subsection{Contract Optimality}
\label{sec:contractoptimality}

To solve for the optimal contract rewards $\tilde{R}^*_i$, we first establish the dependence of optimal contract rewards $\mathbf{{R}}$ on route coverage $\theta$. Thereafter, we solve the problem in (\ref{eqn:contractoptimization}) with $\boldsymbol{\theta}$ only. Specifically, we obtain the optimal rewards $R^*(\boldsymbol{\theta})$ given a set of feasible node coverages from each UAV which satisfies the monotonicity constraint $\theta_{1} \geq \theta_{2} \geq \cdots \geq \theta_{i} \geq \cdots \geq \theta_{YZ}$. 

In addition, the multi-dimensional contract formulation that we have  thus far only considers the self-revelation for two types, i.e., sensing and computation costs. To account for traversal and transmission cost types, we add an additional fixed compensation $\hat{R}$ into the contract rewards.  The traversal cost can be derived from the historical information of the UAV, and can be calibrated based on the response that the model owner receives. In the following theorem, we prove that the addition of a fixed rewards compensation does not violate the IC constraints, i.e., the self-revealing properties of the contract is still preserved, whereas it is inconsequential even if the IR constraint is violated, given that $\tilde{R}$ has already been designed to sufficiently compensate marginal costs, and only one optimal UAV is required to serve each subregion. The optimal rewarding scheme is summarized as follows.

\begin{theorem}
\label{contracttheorem}
\normalfont
For a known set of node coverage $\boldsymbol{\theta}$ satisfying $\theta_{1} \geq \theta_{2} \geq \cdots \geq \theta_{i} \geq \cdots \geq \theta_{YZ}$ in a feasible contract, the optimal reward is given by:
\begin{flalign*} 
R_{i}^{*}=\left\{\begin{array}{c}{\hat{R}-\eta(\theta_{i},\Phi_{i}),   \quad \quad\quad\quad \quad\quad\quad \quad \quad \text { \textit{if} } i= YZ,} \\ 
 {\hat{R}+\tilde{R}_{i+1} + \eta(\theta_{i+1},\Phi_{i})-\eta(\theta_{i},\Phi_{i}),  \text { \textit{otherwise.}}}\end{array}\right.
		\numberthis
\label{eqn:contractheorem}
\end{flalign*}
\end{theorem}

\begin{proof}
There are two parts to the proof. Firstly, we prove that the reward design for the two-dimensional contract is optimal. Therefore, we adopt the proof by contradiction. We first assume there exists some $\mathbf{R}^\dagger$ that yields greater profit for the model owner, meaning that the theorem is incorrect, i.e., $\Pi(R^\dagger)>\Pi(R^*)$. For simplicity, we need to consider only the rewards portion of the model owner's profit function in this proof, i.e., $\sum^{YZ}_{i=1}R^\dagger_i<\sum^{YZ}_{i=1}R^*_i$. This implies there exists at least a $t\in \{1,2, \ldots, YZ \}$ that satisfies the inequality $R^\dagger_t < R^*_t$. 

According to the PIC constraint of Lemma (\ref{lemma:reduceic}), we have:
\begin{flalign}
R_{t}^{\dagger} \geq R_{t+1}^{\dagger}+\eta\left(\theta_{t+1}, \Phi_{t+1}\right)-\eta\left(\theta_{t}, \Phi_{t}\right).
\label{lemma:reduce1}
\end{flalign}
In contrast from Theorem 3, we have:
\begin{flalign}
R_{t}^{*} = R_{t+1}^{*}+\eta\left(\theta_{t+1}, \Phi_{t+1}\right)-\eta\left(\theta_{t}, \Phi_{t}\right).
\label{lemma:reduce2}
\end{flalign}
From (\ref{lemma:reduce1}) and  (\ref{lemma:reduce2}), we can deduce that $R^\dagger_{t+1} < R^*_{t+1}$. Continuing the process up to $t=YZ$, we $R^\dagger_{YZ} \leq R^*_{YZ} = -\eta(\theta_{i},\Phi_{i}) $, which violates the IR constraint. As such, there does not exist the rewards $\mathbf{R}^\dagger$ in the feasible contract that yields greater profit for the model owner. Intuitively, the model owner chooses the lowest reward that satisfies the IR and IC constraints for profit maximization. 

Secondly, we show that adding a fixed traversal cost reward does not violate the IC constraint. Within a subregion, when we consider the complete utility function of the auxiliary UAV with type $i$, $i \neq i'$, 
\begin{flalign*}
\eta\left(\theta_{i}, \Phi_{i}\right)+\tilde{R}_{i} +\hat{R} - & \psi^n_i -\zeta^n_i \geq \\
& \eta\left(\theta_{i'}, \Phi_{i}\right)+\tilde{R}_{i} +\hat{R} - \psi_i^n- \zeta^n_i .
\numberthis
\end{flalign*}
Intuitively, the traversal and transmission cost are structurally separate from the marginal costs, i.e., sensing and computation cost types of the UAV within a subregion $n$. As such, the fixed reward terms cancel out and the self-revealing properties of the contract is preserved.

Note that the IR constraint may no longer hold for some $i$ where
\begin{flalign}
 \psi_i > \eta\left(\theta_{i}, \Phi_{i}\right)+\tilde{R}_{i} +\hat{R}^n .
\numberthis
\label{eqn:luic}
\end{flalign}
However, this is inconsequential given that unlike the conventional contract theoretic formulations, we only require a type of UAV to serve a subregion. Moreover, $\tilde{R}$ is already designed such that the  IR constraints hold to compensate marginal costs sufficiently. For certain subregions, $\hat{R}$ can be calibrated upwards if  no UAV responds to the model owner.

\end{proof}

Following (\ref{eqn:contractheorem}), we can re-express the optimal rewards as:
\begin{flalign}
R_{i}^{*}=\hat{R} -\eta(\theta_{YZ},\Phi_{YZ}) +\sum_{t=i}^{YZ} \Delta_{t},
\end{flalign}
where $\Delta_{YZ} = 0 $,  $\Delta_t = \eta(\theta_{i+1},\Phi_{i})-\eta(\theta_{i},\Phi_{i})$, and $t= 1,2, \ldots,YZ-1$.

Unlike conventional contract theoretic formulations, the model owner only requires a single contract per subregion $n$ for the optimal UAV. From (\ref{eqn:modelownerutility}), we can deduce that the optimal type-${i^*}$ UAV  to serve each subregion is $i^* = \arg \max_{\Phi_i \in \boldsymbol{\Phi}}\Pi(\Omega \{ \mathcal{A},\mathcal{B} \})=\arg \min_{\Phi_i \in \boldsymbol{\Phi}}\Phi_i  = \Phi_1$. Intuitively, for each region, the model owner leverages on the self-revealing properties of the multi-dimensional contract formulation to obtain an optimal UAV with the lowest marginal cost of node coverage for profit maximization. In other words, this is the UAV that can cover the largest proportion of the subregion at the lowest cost, among all feasible UAVs that can complete the task. We can substitute the optimal rewards into the profit function of the model owner and rewrite the profit maximization problem as follows: 
\begin{flalign*}
&\max_{(R,\theta^n_{i^*})}  \Pi(\Omega^n) =  \sum^N_{n=1} G^n (\theta^n_{1}), \\
& \textrm{s.t.} \\
& C1: \text{ } \theta^n_{1} \geq \theta^n_{2} \geq \cdots \geq \theta^n_{i} \geq \cdots \geq \theta^n_{YZ},\\
& C2: \text{ } 0 \leq \theta^n_{1} \leq 1,
\numberthis
\label{contracteqnfin}
\end{flalign*}
where: 
\begin{flalign*}
G^n &=  \frac{\sigma }{N} log(1+ \mu \theta^n_1 D^n)- R^*_1 \\
&=   \frac{\sigma }{N} log(1+ \mu \theta^n_1 D^n ) - \hat{R}- \phi( \alpha^n_1 \theta^n_{1} + \beta_1 \theta^n_{1} ) . 
\end{flalign*}

Note that $C1$ refers to the monotonicity constraint of the contract whereas $C2$ specifies the upper and lower bounds of the route coverage. Taking the first order derivative of  (\ref{contracteqnfin}) to solve for $\theta^{n*}_{1}$, we can derive the closed form solution of the optimal node coverage-contract reward pair as follows:
\begin{flalign*} 
\left\{\begin{array}{c}{\theta^{n*}_{1} = \frac{1}{\mu D^n} \left( \frac{\sigma}{N \phi(\alpha_1+\beta_1)} -1   \right)  ,   } \\ 
\\
{R_{i}^{*}=\hat{R} -\eta(\theta_{YZ},\Phi_{YZ}) +\sum_{t=i}^{YZ} \Delta_{t}}.\end{array}\right .
		\numberthis
\label{eqn:optimalcontractspecs}
\end{flalign*}
Following (\ref{eqn:optimalcontractspecs}), we are able to derive an optimal node coverage-contract reward pair upon the announcement of the UAV types. If the derived contract pairs satisfy the monotonicity conditions, they are the optimal contract formulation. Otherwise, we use the iterative adjusted algorithm, i.e., ``Bunching and Ironing''  algorithm \cite{gao2011spectrum}, to return the results that satisfy the monotonicity constraint.

Note that to derive the optimal rewards for any type-$i^*$ UAV, it is necessary to obtain the rewards for  type-${i^*+1}, \ldots, $ type-$YZ$ UAV. However, since there can only be one selected UAV for each subregion, the reward for the ${(i^*+1)}^{th}$ UAV and onward has to be zero in practice. As such, to compute the rewards, we first assume the type-$YZ$ UAV is selected to obtain a hypothetical $R_{YZ}$. With this, we can then work from backwards, towards deriving $R_{YZ-1}, \ldots , R_{i^*}$. Thereafter, the model owner specifies route assignment $\mathcal{A}^{1, n}$ for the optimal UAV such that $\theta^{n*}_1 =  \frac{\sum_{i' \neq i, i' \in \mathcal{I}_n}a^{1,n}_{i,i'}}{|\mathcal{I}_n|}$. 

In the next section, we consider the UAV-subregion assignment.

\section{UAV-Subregion Assignment}
\label{sec:matching}

In this section, we consider the matching-based UAV-subregion assignment using the GS algorithm. In Section \ref{sec:contract}, we note that the optimal UAV for each subregion has to be the lowest UAV type, i.e.,  type-$1$ UAV. However, given that all subregions will prefer type-$1$ UAV, there exist a need to consider a two-side matching such that the optimal UAVs are matched to the subregions efficiently.

%We denote $\mathcal{P}_j$ as the set of preferences for subregions by the UAV $j$. Given the contract theoretic formulation discussed in Section \ref{sec:contract}, we can deduce that the surpluses from the contractual rewards for sensing and computation are constant across all subregions for the same UAV. On the other hand, the surpluses from sensing and transmission cost compensation differ across subregion. As such, the subregions in the preference set $\mathcal{P}_j$ are arranged in the ascending order based on the value of $\phi_j^n + \zeta_j^n$. In other words, the UAVs prefer subregions closer to them, or with more favourable channel conditions, so as to benefit from a larger surplus $\hat{R} - \phi_j^n - \zeta_j^n$. Similarly, we denote $\mathcal{P}_n$ as the set of preferences for UAVs by the subregion $n$. 

\subsection{Matching Rules}

\begin{algorithm}
	\caption{GS Algorithm for UAV-Subregion Assignment}\label{algorithm} 
	\begin{algorithmic}[1]
	
	\State \textbf{Input:} $\mathcal{N}, \bar{\tau_n}, \mathcal{J}, \Psi, \mathcal{A}, \mathcal{B}, \mathcal{C}, E_j$
	
	\State \textbf{Output:} $\mathcal{M}^*(j), \forall j \in \mathcal{J}$
	
\State \underline{Phase I: Initialization}

\If {$\tau_P^{j,n}(\hat\theta) +\tau^{j,n}_C(\hat\theta) + \tau^{j,n}_T \leq \bar{\tau_n}(\hat\theta)$}

	\State UAV $j \in \mathcal{J}$ report types $\alpha^n_j, \beta^n_j$ to subregion $n$
\EndIf

\State Sort $\Phi(\alpha^n_j, \beta^n_j) $ in an ascending order to derive $\mathcal{P}_n$

\State \textbf{Set} $\mathcal{R} = \mathcal{N}$, $\mathcal{M} =\emptyset$

\State \underline{Phase II: Iterative Matching}

\While {$\mathcal{R} \neq \emptyset$ and $\mathcal{P}_n \neq \emptyset$, $\forall n \in \mathcal{R}$}

	\For {$n$ in $\mathcal{R}$}
				\State $j^*= \arg \min_{\Phi_i \in \boldsymbol{\Phi}}\Phi_i $
				\State Formulate $({\theta^{n}_{j^*}, R^{n}_{j^*}})$ and propose contract

		\While {$n$ has more than one optimal UAV}
		
			\State $\tilde{R}^n = \tilde{R}^n - \Delta \tilde{R}^n$

		\EndWhile

		\If{ $u^n_{j^*}(\theta^{n}_{j^*}) > u^{n^\prime}_{j^*}(\theta^{n^\prime}_{j^*}) $}
				
				\State $\mathcal{M}(j^*) =n$
				\State Remove $n$ from $\mathcal{R}$ and add $n^\prime$ to $\mathcal{R}$

			\Else
			
				\State Remove $j^*$ from $\mathcal{P}_n$
				
			\EndIf

	\EndFor

\EndWhile
			
	\end{algorithmic} 
\end{algorithm}

We introduce a complete, reflexive, and transitive binary preference relation \cite{roth1992two}, i.e., ``$\succ $''  to study the preferences of the UAV. For example,  $n \succ_{j} n^{\prime}$ implies that the UAV $j$ strictly prefers subregion $n$ to $n^\prime$, whereas $n \succeq_{j} n^{\prime}$ indicates that the UAV $j$ prefers the subregion $n$ at least as much as UAV $j^\prime$. We also consider the core definitions as follows:

\begin{definition}
\normalfont
\textit{Matching:} For the formulated matching problem $(\mathcal{J},\mathcal{N},\mathcal{P}_n,\mathcal{P}_j)$, where $\mathcal{J}$ and $\mathcal{N}$ denote the set of UAVs and the set of subregions respectively, whereas  $\mathcal{P}_n$ and $\mathcal{P}_j$ denote the preferences of the subregions and UAVs respectively. The matching $\mathcal{M}(j)=n$ indicates that the UAV $j$ has been matched to subregion $n$, whereas the matching $\mathcal{M}(j)=\emptyset$ implies that the UAV $j$ has not been matched to any subregion.
\end{definition}

\begin{definition}
\normalfont
\textit{Propose Rule:} The UAV $j \in \mathcal{J}$ announces its type to all the eligible subregions based on feasibility of task completion. Then, the contract is formulated and the subregion proposes to its most preferred UAV $j^*$ in its preference set $\mathcal{P}_n$ , i.e., $n^* \succ_{j} n^{*\prime}$, $n^* \forall n \in \mathcal{N}$, $n^* \neq n^{*\prime}$. Note that the preference of the subregion is managed by the FL model owner.
\end{definition}

\begin{definition}
\normalfont
\textit{Reject Rule:} The UAV $j \in \mathcal{J}$ rejects the subregion if a better matching candidate exists. Otherwise, the subregion that is not rejected will be retained as a matching candidate.
\end{definition}

%\begin{definition}
%\normalfont
%\textit{Deferred Acceptance:} The UAVs that have not been rejected are kept as matching candidates for the subregion. The UAVs will subsequently be rejected if another candidate appears in subsequent iterations.
%\end{definition}

However, given that the subregion preference for UAV type is only based on two of four type dimensions, i.e., marginal costs, some subregions may have multiple preferred UAVs with the same marginal costs. To that end, adopting the approach in \cite{zhou2019computation}, we introduce the rewards calibration rule by adjusting the traversal and transmission compensation downward to further reduce the number of eligible UAVs matched to each subregion.

\begin{definition}
\label{rewardscalibration}
\normalfont
\textit{Rewards Calibration Rule:} For the subregion $n \in \mathcal{N}$ that has more than one optimal UAV matched, the contractual rewards can be adjusted downwards, following which the preference of the UAVs are renewed for another iteration of matching. The adjustment is as follows:
\begin{equation}
\tilde{R}^n = \tilde{R}^n - \Delta \tilde{R}^n.
\numberthis
\end{equation}
 
\end{definition}

\subsection{Matching Implementation and Algorithm}

We now explain the implementation procedure of the UAV-subregion assignment. 

\textbf{Phase 1: Initialization}

\begin{itemize}

\item The model owner announces the sensing subregions and time constraint $\bar{\tau_n}$ for task completion to all UAVs in $\mathcal{J}$. Since $\theta^n_j$ is not known apriori, the time limit for task completion can be a function of the variable $\hat{\theta}$ specified by the model owner, e.g., $\hat{\theta} = 0.8$, i.e., only UAVs that can complete 80\% of the route coverage should announce their types (Line 4 of Algorithm 1).

\item The UAVs announce their types $\alpha^n_j, \beta^n_j$ to the feasible subregions (Line 5), i.e., subregions in which the UAVs are able to meet the requirements for task completion. 

\item We initialize (Lines 6, 7): 
\begin{itemize}
	
	\item $\mathcal{M}$ as an empty set
%	\item $\mathcal{T}$ as the set of subregions with more than one optimal UAV
	\item $\mathcal{R}$ as the set of subregions that have yet to be matched, i.e., $\mathcal{R}= \mathcal{N}$ at initialization
	\item $\mathcal{P}_n$ as the set of subregion preferences based on the ascending order of marginal costs  $\Phi(\alpha^n_j, \beta^n_j) $ as discussed in Section \ref{sec:contractoptimality}.

\end{itemize}

\end{itemize}

\textbf{Phase 2: Iterative Matching}

Each iteration of matching consists of the four stages.

\begin{itemize}

	\item \textit{Proposal:} For each subregion $n \in \mathcal{R}$, the node coverage-contract reward pair is formulated for the most optimal UAV $j^*= \arg \min_{\Phi_i \in \boldsymbol{\Phi}}\Phi_i $ in the preference set $\mathcal{P}_n$. Then, the subregion proposes to an optimal UAV (Lines 11-12). 
	
		\item \textit{Rewards Calibration:} If the subregion has more than one optimal UAV, the contract reward is calibrated downward until a one-to-one matching is achieved (Line 13-14).
	
	\item \textit{Rejection:} If the UAV has a better matching candidate $u^n_{j^*}(\theta^{n}_{j^*}) < u^{n^\prime}_{j^*}(\theta^{n^\prime}_{j^*}) $, the subregion is rejected. If not, the UAV keeps the subregion as a matching candidate.

	\item \textit{Update:} If a subregion has been matched, remove it from $\mathcal{R}$ (Line 17). If a prevailing matching candidate has been rejected, add it back to $\mathcal{R}$ (line 17). Update $\mathcal{P}_n$ by removing the UAV that has issued a rejection in this iteration (Line 19).
	
\end{itemize}

The iterations are repeated until all subregions have been matched, i.e., $\mathcal{R} = \emptyset$, or the remaining subregion has been rejected by all UAVs in its preference list, i.e., $\mathcal{P}_n= \emptyset$. The pseudocode is presented in Algorithm \ref{algorithm}.

The stability and optimality properties of the GS algorithm are ensured following the  proofs in \cite{o2007algorithmic}. As such, the self-revealing properties of our multi-dimensional contract design assures truthful type reporting, whereas the matching algorithm ensures that only one optimal UAV is matched to each region. In the following, we perform the performance evaluation of our incentive design.

\section{Performance Evaluation}
\label{sec:pe}

\begin{figure*}
\centering
\begin{multicols}{3}

    \includegraphics[width=\columnwidth]{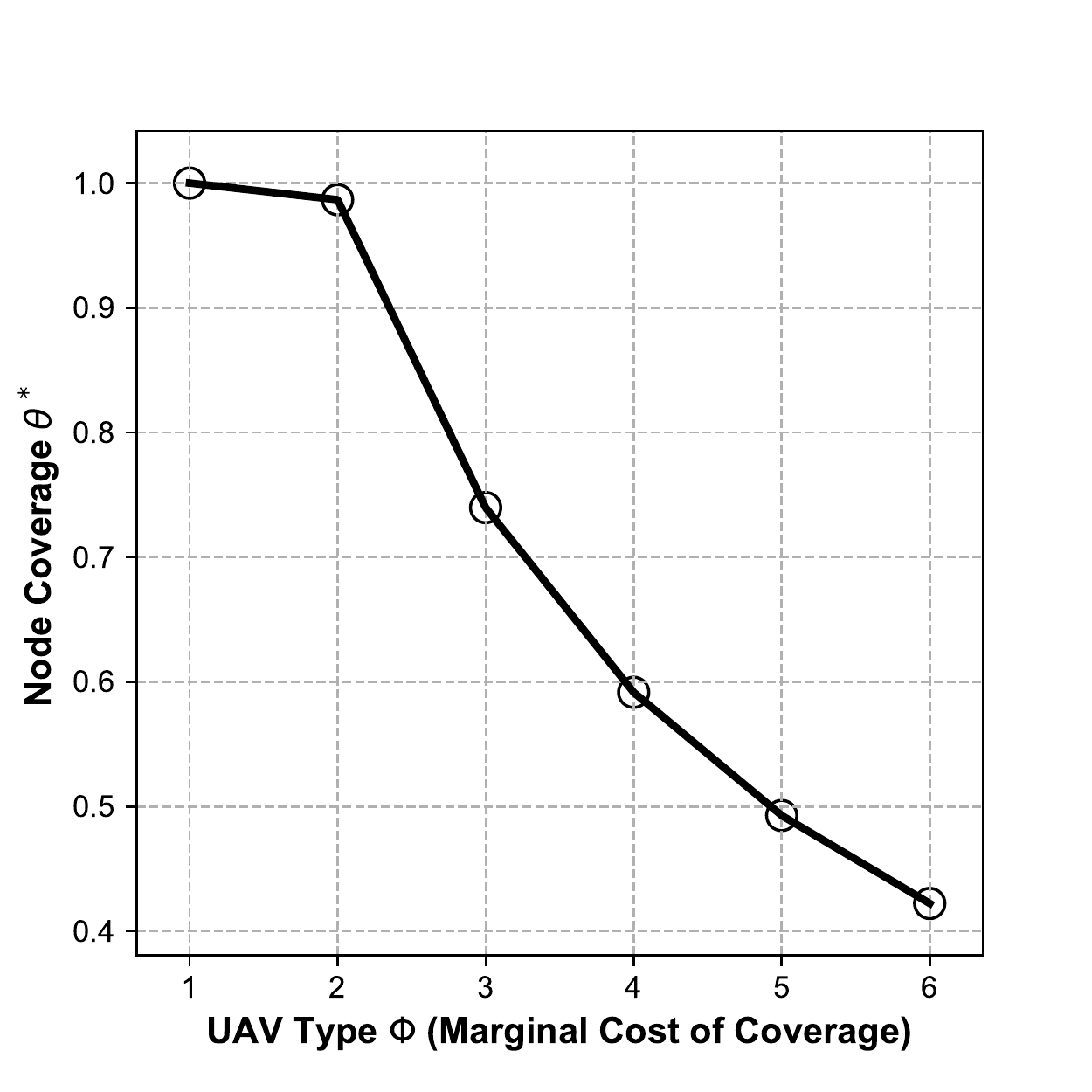}\par 
		\caption{UAV node coverage vs. auxiliary types}
	\label{fig:nodecoverage}
	\includegraphics[width=\columnwidth]{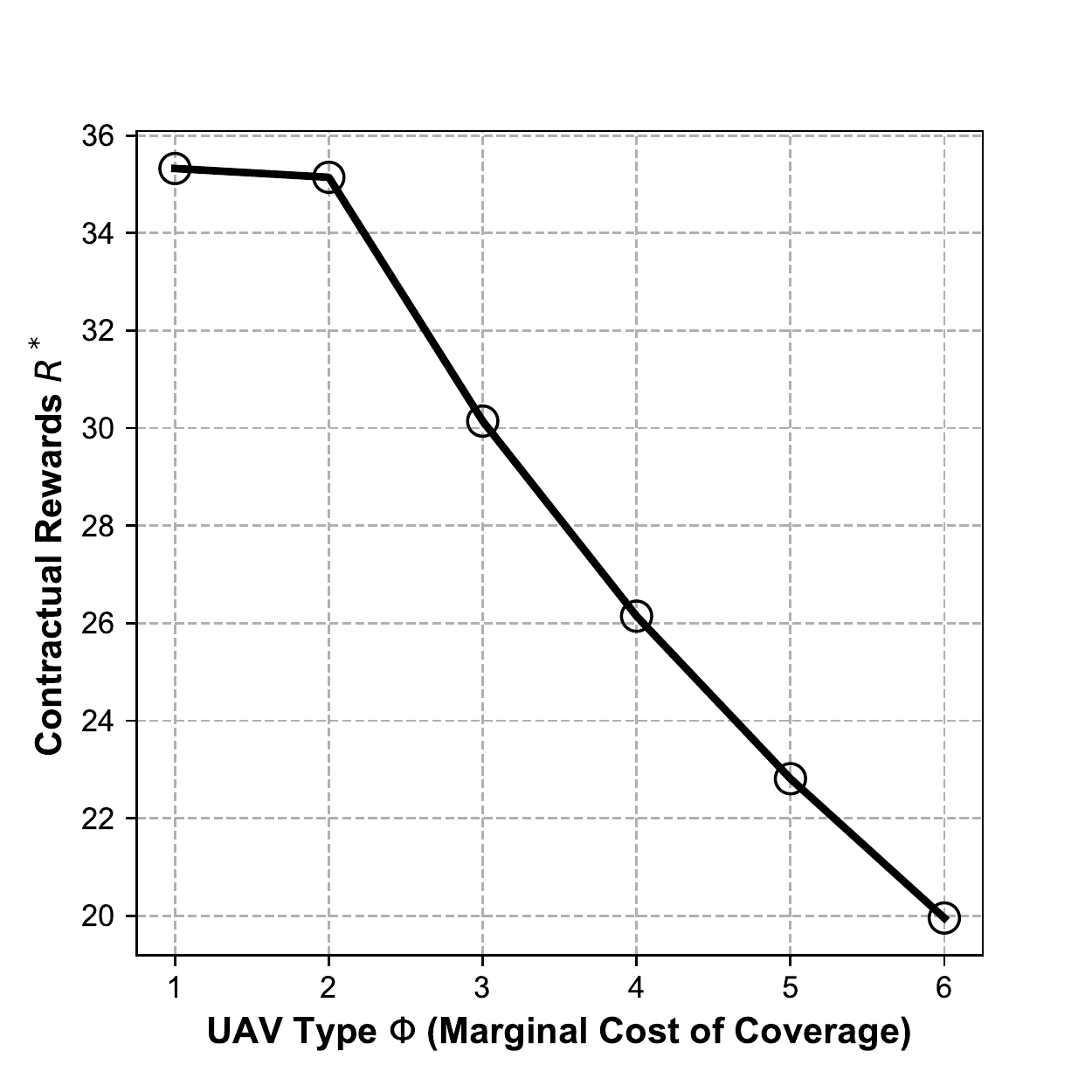}\par 
    \caption{Contract rewards vs. auxiliary types.}
    \label{fig:contractrewards}
    \includegraphics[width=\columnwidth]{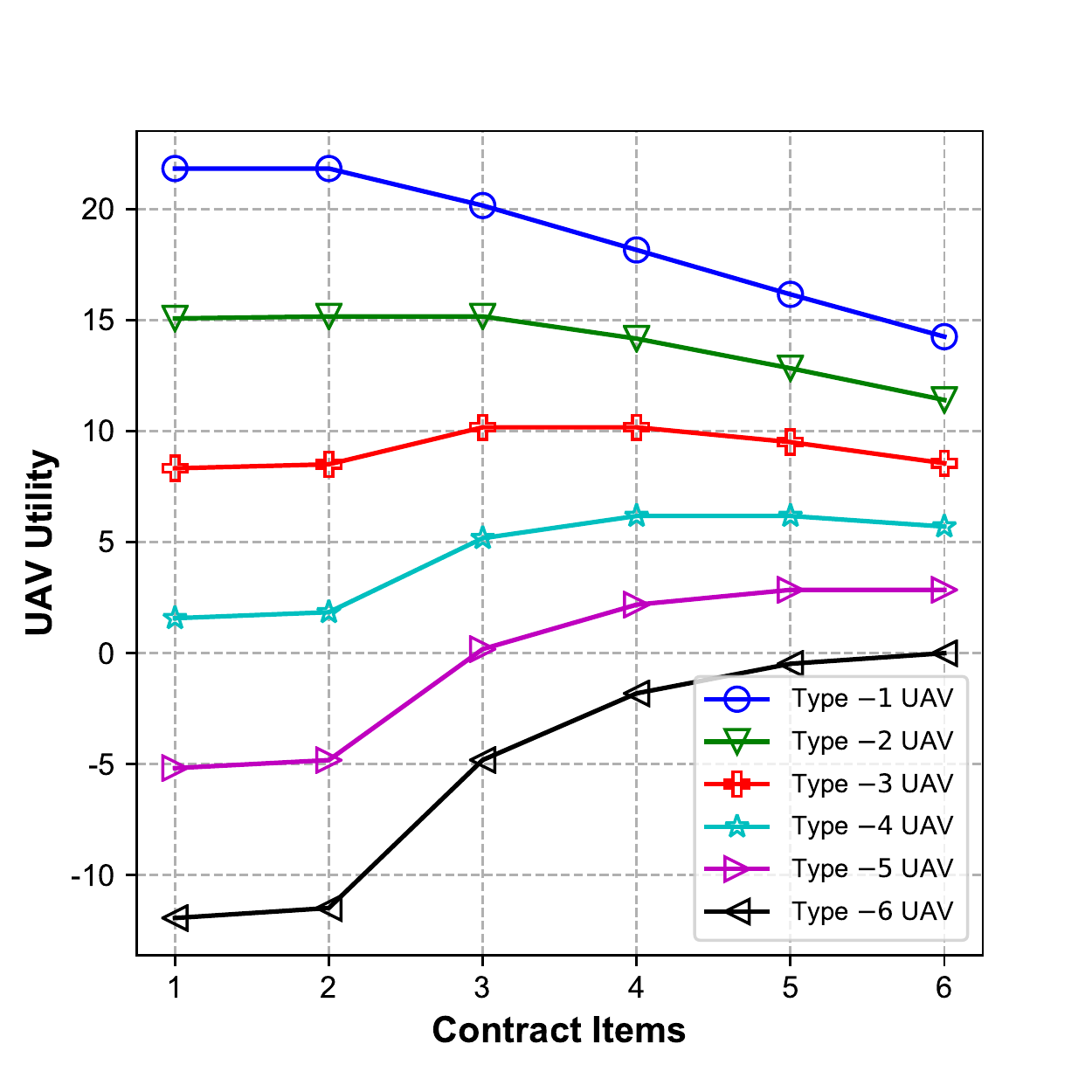}\par 
    \caption{Contract items vs. UAV utilities.}
    \label{fig:contractitems}
    
    \end{multicols}
	
\end{figure*}

\begin{figure}
\centering
	\includegraphics[width=0.6\columnwidth, height= 5cm]{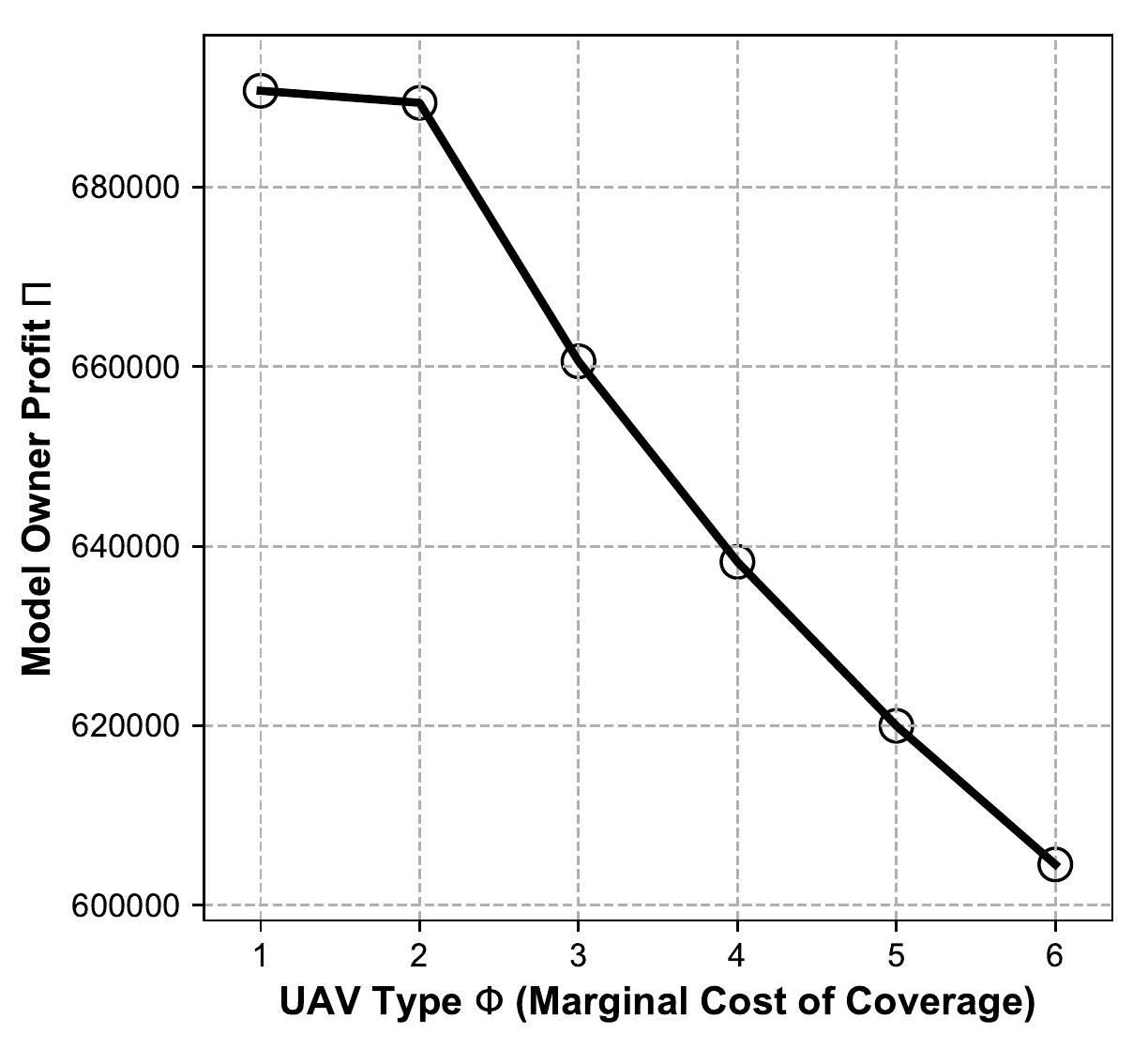}\par 
	\caption{The model owner profits vs. UAV auxiliary types.}
	\label{fig:contractprofit}
\end{figure}

\begin{figure}
\centering
	\includegraphics[width=0.8\columnwidth, height= 5cm]{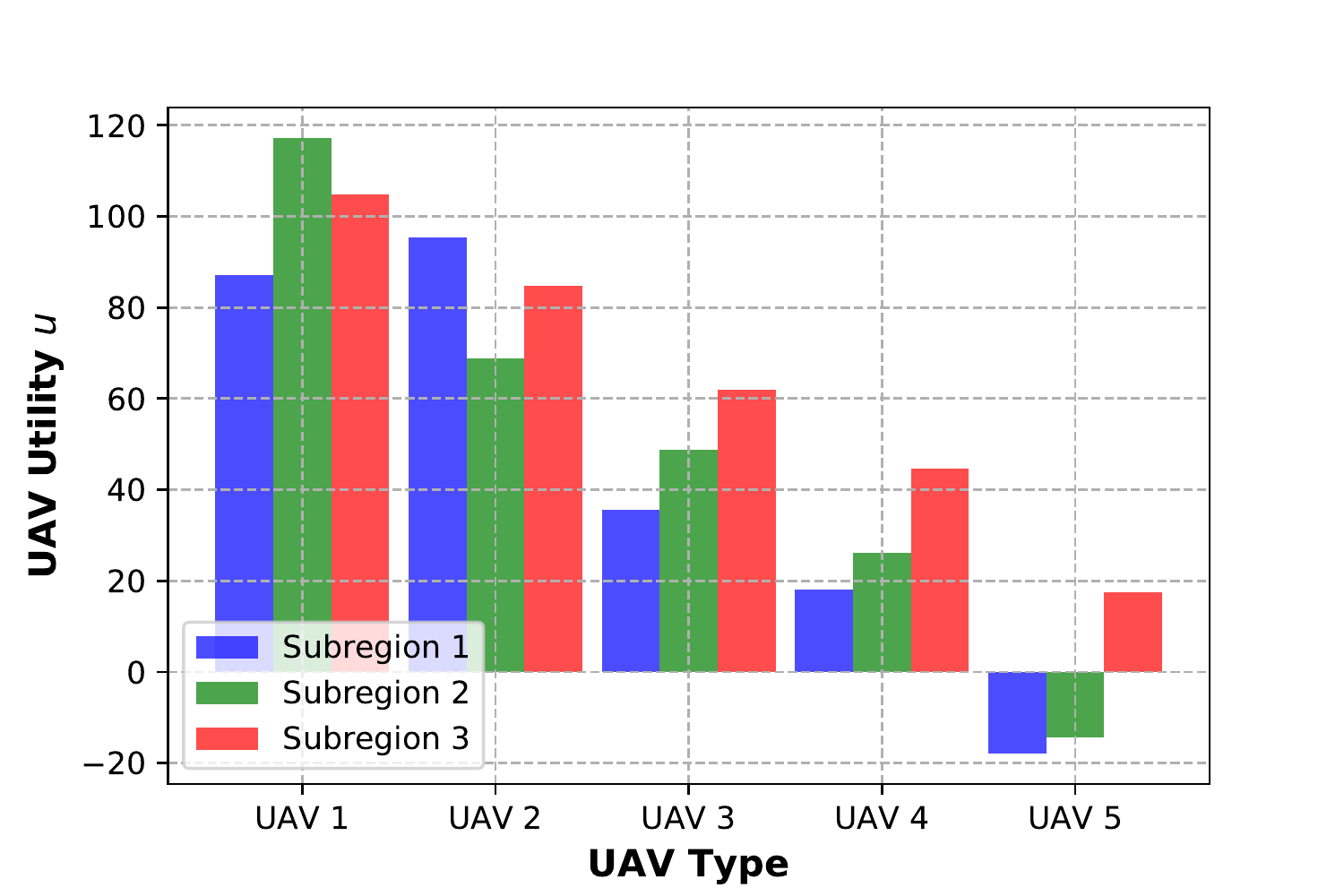}\par 
	\caption{The UAV utility for each subregion vs. types.}
	\label{fig:preferenceanalysis}
\end{figure}

\begin{figure*}
\centering
\begin{multicols}{3}

    \includegraphics[width=\columnwidth]{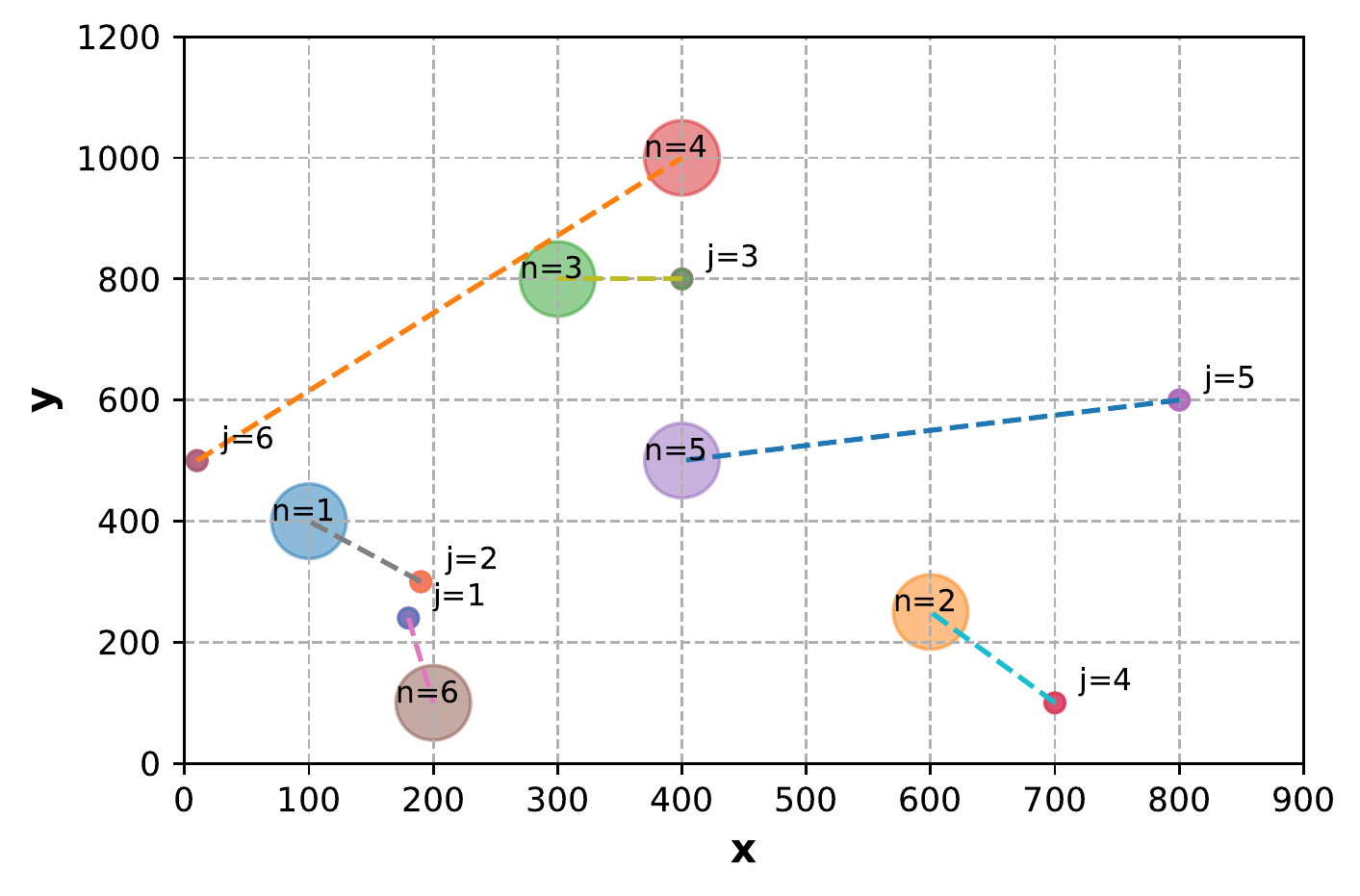}\par 
		\caption{UAV matching for homogeneous subregions.}
	\label{fig:match1}
	\includegraphics[width=\columnwidth]{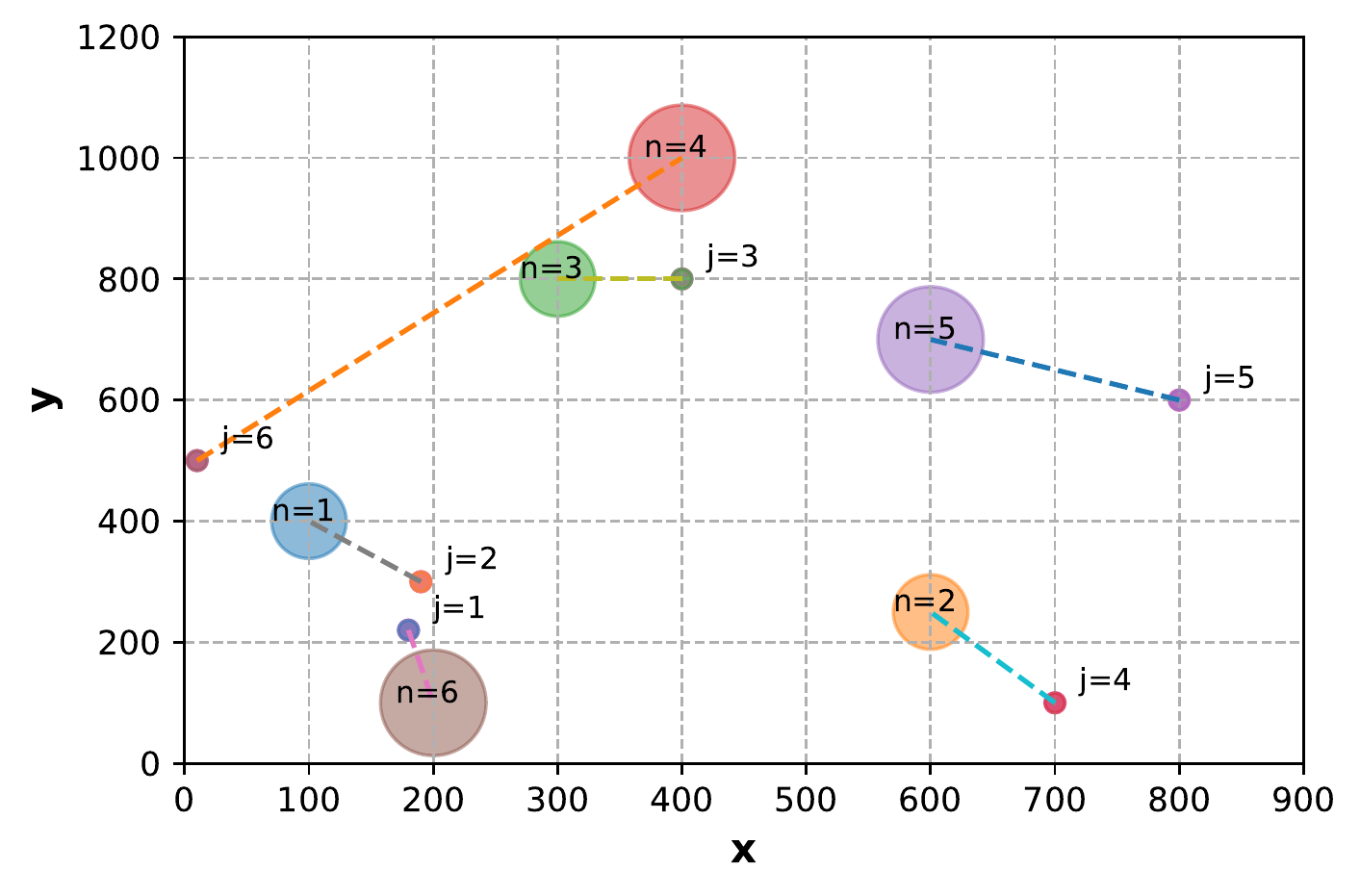}\par 
    \caption{UAV matching for subregions with different data quantities and coverage area (indicated by size of circle).}
    \label{fig:match2}
    \includegraphics[width=\columnwidth]{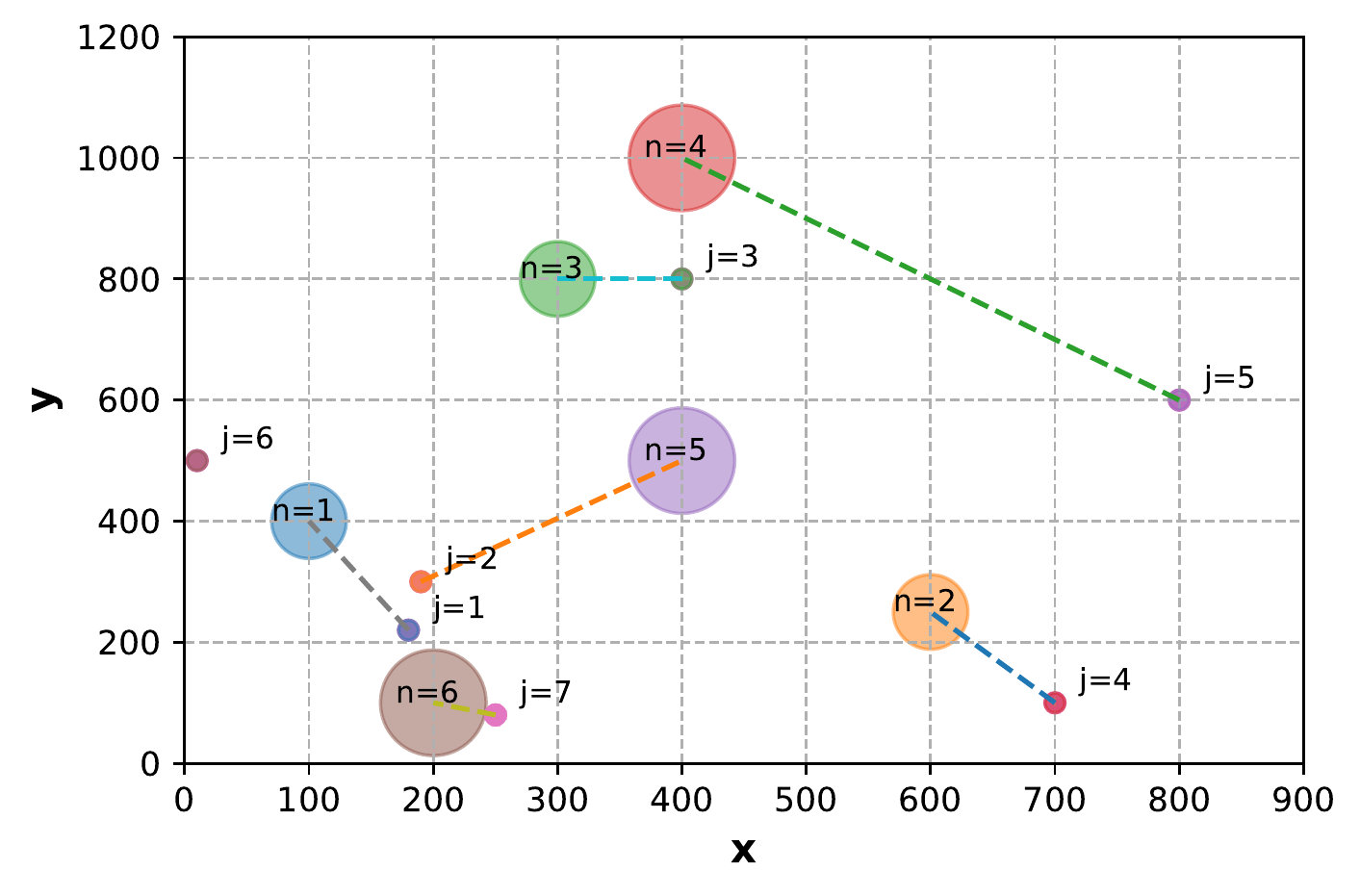}\par 
    \caption{UAV matching where $J>N$.}
    \label{fig:match3}
    
    \end{multicols}
	
\end{figure*}

In this section, we consider the optimality of our devised contract. To illustrate the contract optimality, we first consider the case of six UAVs and a single subregion. Then, we study a single iteration of matching between $5$ UAVs and $3$ subregions. Finally, we study the GS matching-based UAV-subregion assignment that involves up to $7$ UAVs and $6$ subregions. Unless otherwise stated, the list of value ranges for the key simulation parameters are as summarized in Table \ref{tableparameters}. The key parameters we use are with reference to studies involving UAV and FL optimization \cite{zeng2019energy,yang2019energy}.

\begin{table}[]

\centering
\caption{Table of Key Simulation Parameters}
\begin{tabular}{|l|l|}
\hline
\multicolumn{1}{|c|}{\textbf{Simulation Parameters}} & \multicolumn{1}{c|}{\textbf{Value}} \\ \hline
\multicolumn{2}{|c|}{\textit{ UAV Sensing and Traversal Parameters}}                                                          \\ \hline
                             $p$  &        10 - 35        \\ \hline
                             $v$  &         10 - 20    m/s  \\ \hline
                            $l^n$  &         1000 - 2000   m    \\ \hline
                            $l^n_{C_j}$  &         500 - 1000   m    \\ \hline
\multicolumn{2}{|c|}{\textit{UAV Computation Parameters}}                                           \\ \hline
                            
                             $L$  &        4       \\ \hline
                              $\varepsilon$  &         $\frac{1}{3}$       \\ \hline
                              $\delta$  &         $\frac{1}{4}$       \\ \hline
                              $\gamma$  &         $2$       \\ \hline
                              $A^*$  &         $0.6$       \\ \hline
                              $\varepsilon$  &         $\frac{1}{3}$       \\ \hline
                              $C$  &         $10 - 30 $    cycles/bit    \\ \hline
                              $K$  &         $24$       \\ \hline
                              $V$  &         $4$       \\ \hline
                                 $\kappa$  &         $10^{-28}$       \\ \hline
                                  $f$  &     $ 2  $       GHz         \\ \hline
\multicolumn{2}{|c|}{\textit{UAV Transmission Parameters}}                                           \\ \hline
                                 $D^n$  &     $ 500 - 1000  $       MB         \\ \hline
                                 $\lambda^n$  &     $ 10000 - 15000  $              \\ \hline
                                  $H$  &     $ 1  $       MB         \\ \hline
                                    $\rho$  &     $ 8 - 18  $               \\ \hline
 \multicolumn{2}{|c|}{\textit{UAV and Model Owner Utility Parameters}}                                           \\ \hline
                                 $\phi$  &     $ 0.05  $                \\ \hline
                                  $\mu$  &     $ 1  $                \\ \hline
                                   $\sigma$  &     $ 100000  $                \\ \hline

\end{tabular}
\label{tableparameters}
\end{table}

\subsection{Contract Optimality}

To illustrate the optimality of our multi-dimensional contract design, we first consider a highly simplified and demonstrative case of a single subregion and six UAVs of ascending marginal cost of route coverage. The values of $\alpha$ for the UAVs lie in the range of $[ 250,875]$, with increments of $125$, whereas the $\beta$ values lie in the range of $[20,70]$, with increments of $10$. The values are varied as presented in Table \ref{tableparameters}. Then, the auxiliary types are derived following (\ref{eqn:auxiliary}), and arranged in an ascending order for contract derivation. Type-$1$ UAV has the lowest marginal cost of node coverage, whereas  type-$6$ UAV has the highest marginal cost of node coverage. To focus our study on the optimality of our contract design, we hold the traversal and transmission cost types of the UAV constant for now. In addition, we assume that all the UAVs can complete the task within the time constraints.

Fig. \ref{fig:nodecoverage} and Fig. \ref{fig:contractrewards} consider the hypothetical scenarios in which each particular UAV type takes turn to be matched to the subregion. As an illustration, if UAV type-$1$ has been matched to serve the subregion, the optimal node coverage is $1$, whereas the contractual rewards is $35$. In contrast, if UAV type-$6$ is matched, the optimal node coverage is close to $0.4$, whereas the rewards is $20$. From Fig. \ref{fig:nodecoverage} and Fig. \ref{fig:contractrewards}, we can observe that the monotonicity condition discussed in Theorem \ref{monotonicitycondition} holds. In other words, the higher is the marginal cost of node coverage, the lower is the optimal node coverage and contract rewards. 

Fig. \ref{fig:contractitems} shows that the IC constraints of the contract holds. As an illustration, we consider the type-6 UAV, i.e., the UAV with the maximum marginal cost of node coverage. The type-6 UAV derives negative utility if it misreports its type, i.e., to imitate any other lower marginal cost UAV types $1-5$. As discussed in Definition (\ref{def:ic}), each UAV derives the highest utility only if it reports its type truthfully to the model owner. This validates the self-revealing mechanism of our contract. 

Fig. \ref{fig:contractprofit} shows that the model owner profits are the highest when it is able to be matched with the UAV of the lowest marginal cost of node coverage. This validates our discussion in Section \ref{sec:contractoptimality}, and confirms the model owner preference. In other words, among UAVs that are able to complete the task, the model owner prefers the UAV that incurs the lowest cost.

\subsection{UAV-Subregion Preference Analysis}

To analyze the preferences of the UAVs and subregions before we proceed with matching, we consider $5$ UAV types and $2$ subregions. In particular, the auxiliary types of the UAVs are shown in Table \ref{parameters2}. Similarly, the types are derived from the calibration of the parameters listed in Table \ref{tableparameters}. The UAVs are sorted in the ascending order based on marginal cost of node coverage. Besides, we consider three subregions $1,2,3$ of coordinates $(1000,1000)$, $(50,50)$, $(500,500)$. The subregion preference for each UAV is also presented in the last column of Table \ref{parameters2}.

\begin{table}[]
\caption{UAV Types For Preference Analysis.}
\label{parameters2}
\centering
\begin{tabular}{@{}llllll@{}}
\toprule
 UAV  Type   & Coordinates & $\alpha$ & $\beta$  & Subregion Preference & \\ \midrule
 $1$ & $(100,100)$    & $500$      & $20$      & $(2,3,1)$ \\
 $2$ &  $(900,900)$    & $500$    &  $20$     & $(1,3,2)$\\
 $3$ &   $(400,400)$   &   $750$   & $30$    &$(3,2,1)$ \\
 $4$ &   $(450,450)$   &  $750$    &    $30$  &$(3,2,1)$ \\
 $5$ &   $(500,500)$   &  $1000$    &     $40$  &$(3)$\\ \bottomrule
\end{tabular}
\end{table}

Following our discussion in Section \ref{sec:contractoptimality}, the subregions prefer the UAV types with lower marginal cost. Naturally, the preferences for all three subregions are similar as follows: $(1,2,3,4,5)$. Note that the subregions are indifferent between types $1$ and $2$, as well as types $4$ and $5$ given that the pairs have the same marginal cost of node coverage. 

To consider the UAV preferences, we plot the potential profits that each UAV may gain from covering the different subregions in Fig. \ref{fig:preferenceanalysis}. Note that this profit is a hypothetical one in some cases, since the profits can only be realized \textit{if} the UAV has been matched to cover the subregion. However, given that the UAV is not aware if it will be matched to the subregion apriori, the preference list of the UAV can only be constructed with the assumption that it is indeed matched to the subregion. We note that UAV $1$ prefers subregion $2$, whereas UAV $2$ prefers subregion $1$ and so on. Intuitively, the preference for subregions relies on the traversal costs, i.e., the cost of traveling to and from the subregion. As such, the preferences for the UAVs $1$ and $2$ are $(2,3,1)$ and $(1,3,2)$ respectively. 

On the other hand, the UAV $5$ prefers only the closest region $3$, given the potential negative profits derived if it serves the other two subregions, as a result of the high marginal costs incurred for task completion. As such, we are able to derive the matching of (Region $2$, UAV $1$) and (Region $1$, UAV $2$) given that the UAV-subregion preferences match perfectly.

The consideration for subregion $3$ is clearly more challenging than that of $1$ and $2$ given that the subregion is indifferent between the two remaining UAVs $3$ and $4$, and that the UAVs also rank the subregion highest, in terms of preference. To that end, we consider the rewards calibration rule proposed in Definition \ref{rewardscalibration}. The contract rewards are calibrated downwards till a UAV emerges as the only choice left. In this case, after the downward calibration of rewards $\tilde{R}$, UAV $4$ will clearly be matched with subregion $3$, given its close proximity to the subregion.

Through this relatively straightforward example, we are able to derive an insight, i.e., a successful match will have the lowest marginal cost type UAVs matched to the subregion that it is situated closest to. Clearly, Fig. \ref{fig:preferenceanalysis} also validates the efficiency of our incentive mechanism design, i.e., the best available UAV is matched to the respective subregion.

\subsection{Matching-Based UAV-Subregion Assignment}

In this section, we consider the matching-based UAV-subregion assignment. In particular, we consider three scenarios to illustrate the matching-based assignment.

In the \textit{first} scenario, six UAVs of ascending marginal cost types are initialized to choose among six subregions that are of varying distances from each UAV. Each of the subregions is calibrated to hold the same quantities of data ($D^n$) and sensing area ($l^n$) for coverage. For ease of exposition, the UAVs are all able to complete their tasks within their energy capacities and stipulated time constraint. The coordinates of the subregions and UAVs, as well as the matching outcomes, are presented in Fig. \ref{fig:match1}. The preference list of the UAVs are presented in Table~\ref{uavpref}. Note that the preference list of each subregion is simply $(1,2,3,4,5,6)$, i.e., among all feasible UAVs that can cover the subregion within the time and energy constraints, the UAV with the lowest marginal cost is preferred.

From Fig. \ref{fig:match1}, we observe that the UAV $1$ is matched to its most preferred subregion $6$. Though UAV $2$ also prefers subregion $6$, it is unable to be matched to the subregion given that UAV $1$ is higher up on the list of preferences of subregion $6$. As such, UAV $2$ is matched to its second choice. Naturally, the matching between UAV $3$ and subregion $3$, UAV $4$ and subregion $2$, as well as UAV $5$ and subregion $5$ is intuitive, given the unavailability of the other more preferred UAVs for the subregions to match with. We observe that UAV $6$ is finally matched with its fifth choice, given that the UAV $6$ has the lowest priority among subregions.

In the \textit{second} scenario, we consider the same UAV types but with heterogeneous subregions of different data quantities and sensing areas for coverage. As was expected, the sizes of the regions do not affect the matching outcomes and the matching remains the same (Fig. \ref{fig:match2}). This is given that the preference rankings of the subregion and the UAV remain constant. While the varying values of $D^n$ and $l^n$ affects the \textit{magnitude} of UAV types, the \textit{ordering} of the UAV types, and thus their preferences, is retained. This is important to ensure that the monotonicity of our contract design holds across subregions, so as to preserve the contract optimality.

In the \textit{third} scenario, we consider the case where $J>N$, i.e., the number of UAVs exceed that of the number of subregions available. As an illustration, we add in the UAV $7$, which has the lowest marginal cost of node coverage relative to that of the other six available UAVs from the aforementioned scenarios. We observe from Fig. \ref{fig:match3} that the matching outcomes have changed. UAV $7$  is now matched with its most preferred subregion, i.e., subregion $6$, in place of UAV $1$. Naturally, this affects the assignment for the other UAVs. For example, UAV $1$ has to be matched to its second choice now, whereas UAV $2$ has to be matched to its third choice. We observe that the UAV of the largest type, i.e., UAV $6$ is left out of the assignment as a result.

The simulation results allow us to validate the efficiency of our mechanism design. Firstly, the contract design ensures truthful type reporting and the incentive compatibility of our contract is validated. Secondly, with consideration of the preferences, the available UAV with the lowest marginal cost of node coverage is matched to the subregion. This ensures the profit maximization of the model owner.

\begin{table}[]
\caption{UAV Type and Preference for Subregions.}
\label{uavpref}
\centering
\begin{tabular}{@{}llllll@{}}
\toprule
 UAV  Type   & Subregion Preference    \\ \midrule
 $1$ & $(6,1,5,2,3,4)$          \\
 $2$ &  $(6,1,5,3,2,4)$        \\
 $3$ &   $(3,4,5,1,2,6)$        \\
 $4$ &   $(2,5,6,1,3,4)$     \\
 $5$ &   $(2,5,3,4,1,6)$     \\
 $6$ &   $(1,5,3,6,4,2)$     \\  \bottomrule
\end{tabular}
\end{table}

\section{conclusion}
\label{sec:conclude}

In this paper, we have considered an FL based sensing and collaborative learning scheme involving UAVs for applications in the IoV paradigm. Given the incentive mismatches between the UAVs and the model owners, we have  proposed a multi-dimensional contract-matching incentive design such that the  UAV with the lowest marginal cost of node coverage is assigned to each subregion for task completion. For future works, we may consider the adoption of wireless charging techniques with energy harvesting \cite{lu2018wireless} such that the UAVs can perform sensing and model training simultaneously, without the need to return to their bases. In that case, the incentive mechanism design will involve the considerations of one more player type, i.e., the wireless charging service provider.

\bibliographystyle{IEEEtran}
\bibliography{fl-uav}

\end{document}